\documentclass[a4paper,USenglish]{lipics-v2018}

\usepackage[noend]{algpseudocode}

\nolinenumbers 
\hideLIPIcs 

\theoremstyle{plain}
\newtheorem{fact}[theorem]{Fact}
\newtheorem{claim}{Claim}
\numberwithin{claim}{theorem}

\newcommand\numberthis{\addtocounter{equation}{1}\tag{\theequation}}

\title{Efficient Distributed Computation of MIS and Generalized MIS in Linear Hypergraphs}

\author{Fabian Kuhn}{University of Freiburg, Germany}{kuhn@cs.uni-freiburg.de}{}{}

\author{Chaodong Zheng}{State Key Laboratory for Novel Software Technology, Nanjing University, China}{chaodong@nju.edu.cn}{}{}

\authorrunning{F. Kuhn and C. Zheng}

\Copyright{Fabian Kuhn and Chaodong Zheng}

\subjclass{\ccsdesc[500]{Theory of computation~Distributed algorithms}}

\keywords{Maximal independent set, generalized maximal independent set, linear hypergraph, CONGEST model.}

\acknowledgements{We would like to given special thanks to Yitong Yin for initially raising up the idea of GMIS, and for the helpful discussions that bootstrapped this paper. We would also like to thank Yannic Maus for the discussions that led to the idea of strict subhypergraph.}

\begin{document}

\maketitle

\begin{abstract}
Given a graph, a \emph{maximal independent set (MIS)} is a maximal subset of pairwise non-adjacent vertices. Finding an MIS is a fundamental problem in distributed computing. Although the problem is extensively studied  and well understood in simple graphs, our knowledge is still quite limited when solving it in hypergraphs, especially in the distributed CONGEST model. In this paper, we focus on linear hypergraphs---a class of hypergraphs in which any two hyperedges overlap on at most one node.

We first present a randomized algorithm for computing an MIS in linear hypergraphs. It has poly-logarithmic runtime and it works in the CONGEST model. The algorithm uses a network decomposition to achieve fast parallel processing. Within each cluster of the decomposition, we run a distributed variant of a parallel hypergraph MIS algorithm by \L{}uczak and Szyma\'nska.

We then propose the concept of a \emph{generalized maximal independent set (GMIS)} as an extension to the classical MIS in hypergraphs. More specifically, in a GMIS, for each hyperedge $e$ in a hypergraph $\mathcal{H}$, we associate an integer threshold $t_e$ in the range $[1,|e|-1]$, and the goal is to find a maximal subset $\mathcal{I}$ of vertices that do not violate any threshold constraints: $\forall e\in E(\mathcal{H}), |e\cap \mathcal{I}|\leq t_e$. We hope that GMIS might capture a broader class of real-world problems than MIS; we also believe that GMIS is an interesting and challenging symmetry breaking problem on its own.

Our second upper bound result is a distributed algorithm for computing a GMIS in linear hypergraphs, subject to the constraint that the maximum hyperedge size is bounded by some constant. Again, the algorithm has poly-logarithmic runtime and it works in the CONGEST model. It is obtained by generalizing our previous (linear) hypergraph MIS algorithm.
\end{abstract}


\section{Introduction}

A \emph{hypergraph} $\mathcal{H}=(\mathcal{V},\mathcal{E})$ is defined by a set of nodes $\mathcal{V}$ and a set of \emph{hyperedges} $\mathcal{E}$. Unlike simple graphs, a hyperedge in a hypergraph can contain two or more nodes. (In this paper, we ignore hyperedges of size one, as for the problems we consider, these hyperedges can be trivially preprocessed.) The maximum hyperedge size of a hypergraph $\mathcal{H}$ is usually called the \emph{dimension} (or \emph{rank}) of $\mathcal{H}$. As Linial~\cite{linial13} and Kutten et al.~\cite{kutten14} have pointed out, while simple graphs capture \emph{pairwise} interactions well, hypergraphs are ideal for modeling \emph{multi-party} interactions. For example, social networks can contain multiple overlapping groups, each of which has multiple individuals; economic transactions often involve several parties, and each party can participate in several transactions at the same time.

Despite their importance, solving graph-theoretic problems in hypergraphs in a distributed fashion is often highly non-trivial, and usually much less understood than the corresponding problems in simple graphs. Computing a \emph{maximal independent set (MIS)} is a prominent example. For a hypergraph $\mathcal{H}=(\mathcal{V},\mathcal{E})$, an \emph{independent set} $\mathcal{I}$ of $\mathcal{H}$ is a subset of $\mathcal{V}$ such that for each hyperedge in $\mathcal{E}$, at least one node is not in $\mathcal{I}$. An independent set $\mathcal{I}$ is called \emph{maximal} if adding any new node to $\mathcal{I}$ would violate independence. Efficient computation of an MIS is an important problem in distributed computing theory: it is a fundamental symmetry breaking problem; it could also be a key building block for solving many other problems (such as matching and vertex coloring).

Efficient algorithms for computing an MIS in simple graphs have long been known, and improvements are still being made (see, e.g., \cite{alon86,luby86,barenboim16,ghaffari16,censor-hillel17}). In contrast, for nearly three decades, researchers have been seeking a parallel algorithm for computing hypergraph MIS within poly-logarithmic time, yet the answer is still unclear (see, e.g., \cite{karp88,beame90,kelsen92,luczak97,bercea14,harris17}). For distributed message-passing systems,  the hypergraph MIS problem has received much less attention. The two classic computational models to study distributed graph problems in message-passing systems are the LOCAL model and the CONGEST model. In both cases, the network is modeled as an $n$-vertex graph and communication happens in synchronous rounds. In the LOCAL model, the messages exchanged in every round can be of arbitrary size, while in the CONGEST model only messages of size $O(\log n)$ are allowed. Currently, to the best of our knowledge, poly-logarithmic time algorithms for the hypergraph MIS problem only exist in the LOCAL model, or in the CONGEST model if the input hypergraph has constant dimension~\cite{kutten14}.

It is no coincidence that the hypergraph MIS problem has a poly-logarithmic time (randomized) LOCAL solution. As has been made explicit by Ghaffari et al.~\cite{ghaffari17}, so long as a graph problem has a ``sufficiently local'' sequential greedy algorithm, there exists a systematical way to build a randomized LOCAL algorithm that solves the problem in poly-logarithmic time. However, this strategy has two drawbacks: (a) large message sizes; and (b) the considered problem is actually solved in a somewhat centralized fashion (though at smaller scale) which might involve non-trivial local computation. On the contrary, to compute an MIS in simple graphs, both the classical algorithm by Luby and Alon et al.~\cite{alon86,luby86} and (the first part of) the latest solution proposed by Ghaffari~\cite{ghaffari16} work well in the CONGEST model, and incur little local computation. Therefore, an interesting open question is: do poly-logarithmic time CONGEST algorithms exist that can solve the hypergraph MIS problem?

In this paper, we make some progress towards answering this open question. Particularly, we focus on \emph{linear hypergraphs}---a class of hypergraphs in which any two hyperedges intersect in at most one node---and devise efficient algorithms to compute MIS and another closely related structure in such hypergraphs, in the distributed CONGEST model. We note that although linear hypergraphs are a specific subclass of hypergraphs, unique challenges that do not arose in simple graphs persist. In general, our hope is that understanding the MIS problem for linear hypergraphs will be an important intermediate step along the path for solving MIS in general hypergraphs, in the distributed CONGEST model.

\bigskip\noindent{\bf MIS in Linear Hypergraphs.} Our first result is a randomized algorithm that computes an MIS for a linear hypergraph in poly-logarithmic time in the distributed CONGEST model. Conceptually, the algorithm contains two parts. In the first part, we utilize \emph{network decomposition}~\cite{elkin16} to decompose the input hypergraph into multiple smaller ones, each with bounded diameter. (The motivation for doing so will be discussed shortly.) The second part contains multiple iterations. In each iteration, within the bounded-diameter subhypergraphs, we further generate \emph{equitable subhypergraphs}. (Roughly speaking, an equitable hypergraph is somewhat like a ``regular graph'' in the simple graph world.) Then, within each equitable subhypergraph, we independently mark each node with a carefully chosen probability, and let marked nodes that do not violate independence constraints join the MIS. Since the subhypergraphs are equitable, we prove that many nodes will decide in each iteration. Hence, after not too many iterations, the algorithm will output a complete MIS.

The second part of this algorithm can be seen as a distributed variant of a parallel hypergraph MIS algorithm proposed by \L{}uczak and Szyma\'nska~\cite{luczak97}. Nonetheless, to maintain the correctness and efficiency of the original algorithm, the conversion process is nontrivial. Particularly, the first issue is that the original algorithm depends on knowledge of some global parameters. To avoid incurring $\Omega(D)$ time complexity where $D$ is the diameter of the input hypergraph, we employ network decomposition. This is also the motivation for the first part of our algorithm. The second and more critical issue is that the original algorithm  depends on $\Theta(n)$ global parameters, in the worst case. These information would cost too much time to collect, even after decomposition. To resolve this problem, we have refined the detailed analysis so that our algorithm now only depends on $O(\log{n})$ parameters.

\bigskip\noindent{\bf The Generalized MIS (GMIS) Problem.} One way to interpret the hypergraph MIS problem is: for each hyperedge $e\in\mathcal{E}$, associate a threshold $t_e=|e|-1$, then an MIS $\mathcal{I}$ is a maximal subset of $\mathcal{V}$ such that for each hyperedge $e\in\mathcal{E}$, the number of nodes in $\mathcal{I}$ does not exceed $t_e$. (I.e., $\forall e\in\mathcal{E}, |\mathcal{I}\cap e|\leq t_e=|e|-1$.) Now, by allowing $t_e$ to be any integer value between one and $|e|-1$, we obtain what we define as the \emph{generalized maximal independent set (GMIS)} problem. That is, in the GMIS problem, for each hyperedge $e\in\mathcal{E}$, we define (as input) a threshold $t_e$ where $1\leq t_e\leq |e|-1$, and the goal is to find a maximal subset $\mathcal{I}$ of $\mathcal{V}$ such that for each hyperedge $e\in\mathcal{E}$, we have $|\mathcal{I}\cap e|\leq t_e$.

As previously mentioned, hypergraphs is an ideal structure to capture multi-party interactions. The thresholds on hyperedges can be used to represent the constraints posed by various problems. Therefore, we believe the additional flexibility of GMIS (in comparison with MIS) would allow it to model a wider range of real-world problems.

\bigskip\noindent{\bf GMIS in Linear Hypergraphs.} Allowing arbitrary thresholds on hyperedges makes the already hard hypergraph MIS problem even more challenging. For example, many hypergraph MIS algorithms critically rely on the property that an independent set in a subhypergraph is also an independent set in the original hypergraph. However, as we shall see, a generalized independent set in a subhypergraph is \emph{not} necessarily a generalized independent set in the original hypergraph. As a result, we might have to adjust the definition of subhypergraph accordingly, which in turn could significantly affect the performance and/or correctness of the original algorithm.

In this paper, we show that GMIS can be solved in $O(\log^2{n})$ time in the LOCAL model. Moreover, by generalizing our previous hypergraph MIS algorithm, we are able to devise a CONGEST algorithm that can solve GMIS in poly-logarithmic time, subject to the constraint that the input hypergraph is linear and has constant dimension. It is also worth noting, although we use the same high-level strategy, important adjustments to both the algorithm and the analysis are made during the generalization process. 

At first glimpse, it may seem easy to obtain a poly-logarithmic time GMIS algorithm for constant dimension hypergraphs, even in the CONGEST model. However, it turns out that the most intuitive strategies do not lead to the desired outcome. For instance, the approach of reducing the maximum hyperedge threshold one by one can be slow. This is because, in the simple graph setting, Luby's algorithm and its variants achieve high efficiency by considering both the nodes that decide to join and not to join the MIS. Yet, for hypergraph GMIS (as well as MIS), it is hard to analyze how many nodes will decide to not join, thus raising difficulties to arguing how fast nodes are removed, or how fast the maximum hyperedge threshold is reduced.

\section{Related Work}

Efficient computation of MIS in simple graphs has always attracted numerous attention. In two seminal papers, Alon, Babai, and Itai, as well as Luby~\cite{alon86,luby86} provided a randomized algorithm which solves the problem in $O(\log{n})$ time. Since then, many other solutions were proposed (see, e.g., Section 1.1 of \cite{barenboim16} for a brief survey), and the current best known (randomized LOCAL) algorithm is proposed by Ghaffari~\cite{ghaffari16}.

Perhaps surprisingly, however, how to efficiently compute MIS in hypergraphs is much less well understood. As we have mentioned earlier, researchers have been seeking a parallel algorithm that can compute a hypergraph MIS within poly-logarithmic time under the PRAM model for decades, yet the answer is still unclear. More specifically, in 1990, Beame and Luby~\cite{beame90} introduced a randomized algorithm with poly-logarithmic runtime for computing an MIS in hypergraphs of dimension three. Kelsen~\cite{kelsen92} improved the analysis of \cite{beame90} so that the algorithm can work for all constant dimension hypergraphs. Later, \L{}uczak and Szyma\'nska~\cite{luczak97} showed that for all linear hypergraphs, the problem can also be solved within poly-logarithmic time. The second part of our hypergraph MIS algorithm is a refined distributed variant of \L{}uczak and Szyma\'nska's algorithm. On the other hand, for general hypergraphs, early result by Karp et al.~\cite{karp88} proved an MIS can be obtained in $O(\sqrt{n}\cdot(\log{n}+\log{m}))$ time where $m$ is the number of hyperedges. Later, by repeatedly using the algorithm of \cite{beame90}, Bercea et al.~\cite{bercea14} gave an algorithm that works in $n^{o(1)}$ time, subject to the constraint that there are not too many hyperedges. More recently, Harris~\cite{harris17} improved the result of Kelsen~\cite{kelsen92} and devised an algorithm with runtime $O(\log^{2^r}{n})$ for hypergraphs with dimension $r$. Lastly, we note that in the original paper by Beame and Luby~\cite{beame90}, the authors also proposed another simple parallel algorithm and conjectured it can solve MIS within poly-logarithmic time, for any hypergraph. However, to the best of our knowledge, the correctness of this conjecture is still unknown.

For message-passing distributed systems, even fewer attention were paid to the hypergraph MIS problem, and the most recent result comes from Kutten et al.~\cite{kutten14}. In their paper, by employing network decomposition~\cite{linial93} and exploiting the local nature of the MIS problem, the authors provided a $O(\log^2{n})$ time LOCAL algorithm. In contrast, under the CONGEST model in which each message is of bounded size, the authors presented two other results: (a) for hypergraphs with constant dimension $d$, a $O(\log^{(d+4)!+4}{n})$ time algorithm; and (b) for general hypergraphs, a $O(\min\{\Delta^\epsilon\log^{(1/\epsilon)^{O(1/\epsilon)}}{n},\sqrt{n}\})$ time algorithm where $\Delta$ is the maximum degree and $1\geq\epsilon\geq 1/(\frac{\log\log{n}}{c\log\log\log{n}}-1)$ for some constant $c$. In our linear hypergraph MIS algorithm, the dimension can be arbitrary, and the degree of the poly-logarithmic term (in the running time) does not depend on the dimension.

\section{Model and Problem}

A \emph{hypergraph} $\mathcal{H}=(\mathcal{V},\mathcal{E})$ is defined by a set of nodes $\mathcal{V}$, and a set of \emph{hyperedges} $\mathcal{E}$. We usually assume $|\mathcal{V}|=n$, and each node has a unique identity. For each hyperedge $e\in\mathcal{E}$, it contains two or more nodes in $\mathcal{V}$. The maximum size of all hyperedges is called the \emph{dimension} (or \emph{rank}) of a hypergraph. A hypergraph is a \emph{linear hypergraph} if for each pair of hyperedges, they overlap on at most one node. For a set of nodes $\mathcal{V}'\subseteq\mathcal{V}$, define $\mathcal{H}'=(\mathcal{V}',\mathcal{E}')$ to be the induced \emph{subhypergraph} of $\mathcal{H}$ where $\mathcal{E}'=\{e\ |\ e\in\mathcal{E},e\subseteq\mathcal{V}'\}$.

For each hyperedge $e\in\mathcal{E}$, we associate an integer threshold $t_e$ where $1\leq t_e\leq |e|-1$. For a subset $\mathcal{I}$ of $\mathcal{V}$, we call it a \emph{generalized independent set} if for each $e\in\mathcal{E}$, $|e\cap\mathcal{I}|\leq t_e$. We say a generalized independent set $\mathcal{I}$ is a \emph{generalized maximal independent set (GMIS)} if adding any extra node to $\mathcal{I}$ would violate some hyperedge's threshold constraint. Notice, if for each hyperedge $e\in\mathcal{E}$ we define $t_e=|e|-1$, then a generalized independent set becomes a classical hypergraph \emph{independent set}, and a generalized maximal independent set becomes a classical hypergraph \emph{maximal independent set (MIS)}.

To model a hypergraph $\mathcal{H}=(\mathcal{V},\mathcal{E})$, we consider a synchronous message-passing network in which time is divided into discrete \emph{slots}. We adopt the \emph{server-client} model used in \cite{kutten14}. In this model, $\mathcal{H}$ is realized as a simple bipartite graph $G_{\mathcal{H}}=(V_{\mathcal{H}},E_{\mathcal{H}})$. The nodes in $G_{\mathcal{H}}$ are partitioned into two sets: $S_{\mathcal{H}}$ and $C_{\mathcal{H}}$. Each node in $S_{\mathcal{H}}$ represents a particular node in $\mathcal{V}$, and each node in $C_{\mathcal{H}}$ represents a particular hyperedge in $\mathcal{E}$. We call the nodes in $S_{\mathcal{H}}$ as \emph{servers}, and the nodes in $C_{\mathcal{H}}$ as \emph{clients}. For a node $u\in C_{\mathcal{H}}$ and a node $v\in S_{\mathcal{H}}$, there is an edge (i.e., a bidirectional communication link) connecting them if and only if the node represented by $v$ is contained within the hyperedge represented by $u$.

Another model to represent a hypergraph $\mathcal{H}=(\mathcal{V},\mathcal{E})$ is called the \emph{vertex-centric} model. In this model, $\mathcal{H}$ is again realized as a simple graph $G_{\mathcal{H}}=(V_{\mathcal{H}},E_{\mathcal{H}})$. However, here $V_{\mathcal{H}}$ simply denotes the set of nodes in $\mathcal{V}$, and there is an edge between two nodes $u$ and $v$ iff there is a hyperedge in $\mathcal{E}$ containing both $u$ and $v$. We call $G_{\mathcal{H}}$ as the \emph{server graph} of $\mathcal{H}$.

Throughout this paper, at the network layer, we use the server-client model to represent hypergraphs. However, for the ease of presentation, we will sometimes discuss the server graph of the specified hypergraph.

Regarding the capacity of the communication links, we will mostly consider the \emph{CONGEST} model. More specifically, in each time slot, for each direction of each link, only a $O(\log{n})$-sized message can be sent. Sometimes, we will also discuss the implications of our results under the \emph{LOCAL} model. In that case, in each time slot, for each direction of each link, an arbitrarily large message can be sent.

In this paper, we are interested in finding efficient distributed algorithms that can solve MIS and GMIS in linear hypergraphs in the CONGEST model. Particularly, we will develop (Monte Carlo) randomized algorithms that can solve the considered problems \emph{with high probability (w.h.p.)}, i.e., a probability that is at least $1-1/n^c$ for some constant $c\geq 1$.

\section{Decomposing Hypergraphs}

Network decomposition (see, e.g., \cite{linial93,elkin16}) is a widely used technique in distributed computing for solving graph theoretic problems. For a simple graph $G=(V,E)$, a \emph{$(d,c)$-network-decomposition} is a partition of $V$ so that: (a) for each slice of the partition (i.e., a subset of $V$), the induced subgraph has diameter at most $d$; and (b) we can assign each slice of the partition a color within a set of $c$ colors, and ensure any two adjacent nodes in $G$ of the same color must be in the same slice of the partition. Moreover, $d$ is called the \emph{weak diameter} if, when computing the diameters of the induced subgraphs, edges not in the subgraph (but in $E$) can be used; otherwise, $d$ is called the \emph{strong diameter}.

For many network algorithms, network decomposition can be used to boost efficiency as it allows for parallelism: subgraphs with the same color can usually be processed at the same time without interfering each other. Network decomposition is also helpful in that it bounds the diameter of the graph instances the algorithm will process.

Our hypergraph MIS/GMIS algorithm also relies on network decomposition to achieve high efficiency: first, decompose the input hypergraph into multiple subhypergraphs of bounded diameter; then, iterate through all colors and run the core MIS/GMIS algorithm in parallel within subhypergraphs of the same color; finally, combine all partial solutions to obtain a complete MIS/GMIS of the original hypergraph.

In this part of the paper, we will present the guarantees provided by the decomposition procedure; we will also show that combining the MIS/GMIS found in each decomposed subhypergraph correctly gives a complete MIS/GMIS of the original input hypergraph.

We begin with the decomposition procedure. The idea of decomposing input hypergraph into multiple smaller ones and then compute MIS for these subhypergraphs in parallel has been used by Kutten et al.~\cite{kutten14}. In that paper, the authors utilized the classical  $(O(\log{n}),O(\log{n}))$-network-decomposition algorithm developed by Linial and Saks~\cite{linial93}. However, Linial and Saks's algorithm only produces a decomposition with weak diameter $O(\log{n})$, thus might result in congestions when communication occurs in multiple subhypergraphs simultaneously. To resolve this issue, Kutten et al.\ slightly modified Linial and Saks's algorithm so as to upper bound the potential congestion.

Recently, Elkin and Neiman developed a new $(O(\log{n}),O(\log{n}))$-network-decomposition algorithm with strong diameter $O(\log{n})$~\cite{elkin16}. This is a strict improvement when compared with Linial and Saks's algorithm. Therefore, we implement Elkin and Neiman's algorithm in our model in this paper. More specifically, the decomposition procedure---which is described in detail in the proof of the following lemma---guarantees the following properties.

\begin{lemma}\label{lemma-decomp-alg}
Let $G_{\mathcal{H}}$ be the server graph of an $n$-node hypergraph $\mathcal{H}$. With high probability, in $O(\log^2{n})$ time slots, for some positive integer $k$, we can partition nodes of $\mathcal{H}$ into $k$ sets $S_{1}, S_{2}, \cdots, S_{k}$, produce $k$ subgraphs of $G_{\mathcal{H}}$ denoted by $G_{1}, G_{2}, \cdots, G_{k}$, and assign a color within a set of $O(\log{n})$ colors to each set, such that: (a) for all $i$, subgraph $G_{i}$ is the induced subgraph of $S_{i}$ and has strong diameter $O(\log{n})$; (b) for any $S_{i}$ and $S_{j}$ that are assigned with the same color, there is no hyperedge in $\mathcal{H}$ that contains nodes in both $S_{i}$ and $S_{j}$.
\end{lemma}

\begin{proof}
We first briefly describe Elkin and Neiman's network decomposition algorithm. (More details can be found in the original paper~\cite{elkin16}.) The algorithm contains $\ln{n}$ stages. In the $i$\textsuperscript{th} stage, there are $2(cn/e^i)^{1/m}$ phases; we also fix $\beta_i=\ln{(cn/e^i)}/m$. Here, $c$ and $m$ are parameters that can be adjusted. Let $G'_1=G_{\mathcal{H}}$. In each phase $t$, we carve a block $W_t$ out of the current graph $G'_{t}$, and let $G'_{t+1}=G'_t\backslash W_t$. Notice, all nodes in $W_t$ gets a unique color, and each connected component in $W_t$ is a slice of the final partition.

In the $t$\textsuperscript{th} phase, each node $v$ in $G'_t$ independently samples a value $r_v$ from the exponential distribution with parameter $\beta_t$, where $\beta_t$ is the value of $\beta$ for the stage the $t$\textsuperscript{th} phase is contained within. Each node $v$ in $G'_t$ broadcasts $r_v$ to all nodes in $G'_t$ that are within distance $\lfloor r_v\rfloor$ from it. On the other hand, each node $y$ in $G'_t$ also records the values of $r_v$ that have reached it, along with the distances to these nodes. Then, $y$ sorts these nodes $v_1,v_2,\cdots,v_x$ according to $g_i=r_{v_i}-\texttt{dist}_{G'_{t}}(y,v_i)$ in decreasing order. Finally, $y$ is added to $W_t$ iff $g_1-g_2>1$.

As have been shown in \cite{elkin16}, by choosing proper $c$ and $m$, w.h.p.\ the above algorithm finishes within $O(\log{n})$ phases, and $r_v$ will always be bounded by $O(\log{n})$. Moreover, any connected component in any $W_t$ has strong diameter $O(\log{n})$. I.e., the algorithm can create a decomposition with strong diameter $O(\log{n})$ in $O(\log^2{n})$ time, using $O(\log{n})$ colors.

We now describe one simple way to simulate Elkin and Neiman's algorithm in our CONGEST server-client model. To implement the $t$\textsuperscript{th} phase, we use $2j=\Theta(\log{n})$ time slots. More specifically, in the first slot within the phase, each server node $v$ sends $r_v$ to its neighboring client nodes. Then, we repeat the following for $2j-1$ time slots: in each even (resp., odd) slot, each client (resp., server) node sends the two maximum $g_i$ it has seen since the beginning of this phase to its neighbors.

To see the correctness of the above simulation, consider a server node $y$. Assume in the original algorithm, in the phase in which $y$ gets a color, the two maximum values it obtained are $g_1=r_{u_1}-\texttt{dist}(y,u_1)$ and $g_2=r_{u_2}-\texttt{dist}(y,u_2)$. Further assume in our simulation, the two maximum values $y$ obtained are $g'_1=r_{u'_1}-\texttt{dist}(y,u'_1)$ and $g'_2=r_{u'_2}-\texttt{dist}(y,u'_2)$. Clearly, $g'_1\leq g_1$. Moreover, in that phase, a value equal to $g_1$ will reach $y$. Otherwise, there must exist value $g''_1=r_{u''_1}-\texttt{dist}(y,u''_1)>g_1$, such that on the path from $u_1$ to $y$, some node $z$ receives both $g''_1$ and $g_1$, and decides stop forwarding $g_1$. In such case, in the original execution, $g''_1$ will reach $y$ as well because $\lfloor r_{u''_1}-\texttt{dist}(z,u''_1)\rfloor\geq\lfloor r_{u_1}-\texttt{dist}(z,u_1)\rfloor$ decides the number of remaining hops the message will propagate (from node $z$), contradicting the assumption that the maximum value received by $y$ is $g_1$. Hence, we know $g'_1=g_1$. Similarly, we can also prove $g_2=g'_2$. Therefore, we know our simulation is correct.

Since simulating one phase costs $\Theta(\log{n})$ time slots (as $r_v\in O(\log{n})$), and there are $O(\log{n})$ phases, the decomposition procedure terminates in $O(\log^2{n})$ time in our network model. The properties in the lemma follow by the definition of network decomposition.
\end{proof}

With a proper decomposition, the core MIS/GMIS algorithm only needs to deal with bounded diameter hypergraphs. In particular, the following lemma---which is inspired by Lemma 3 in \cite{kutten14}---shows that if we can compute MIS/GMIS in low diameter hypergraphs fast, then we can also compute it in general hypergraphs fast. Notice, when compared with the original version, the proof is generalized so that the claim holds for GMIS as well.

\begin{lemma}[Decomposition lemma, generalized version of Lemma 3 in \cite{kutten14}]\label{lemma-decomp}
Assume we are given a hypergraph $\mathcal{H}$ containing $n$ nodes. If there exists an algorithm $\mathcal{A}$ that computes an MIS (resp., GMIS) for hypergraph $\mathcal{H}'$---which contains $n'\leq n$ nodes and has $O(\log{n})$ diameter---in $T(n')$ time, then there exists an algorithm that computes an MIS (resp., GMIS) for $\mathcal{H}$ within $O(T(n)\cdot\log{n}+\log^2{n})$ time.
\end{lemma}

\begin{proof}
Let $G_{\mathcal{H}}$ be the server graph of $\mathcal{H}$. First, run the network decomposition algorithm on $G_{\mathcal{H}}$ as discussed in the proof of Lemma \ref{lemma-decomp-alg}. This step takes $O(\log^2{n})$ time slots.

The next step contains $O(\log{n})$ iterations, and in the $i$\textsuperscript{th} iteration we consider node sets with color $i$. Assume node set $S_t$ has color $i$, and the corresponding subgraph is $G_{t}$. In the $i$\textsuperscript{th} iteration, we need to decide for each node in $S_{t}$ whether it is in the final solution of MIS (resp., GMIS) or not. In the following analysis, we assume we have already done so for the node sets with color $1$ to $i-1$.

Define $\mathcal{H}_{t}$ to be the following subhypergraph. $\mathcal{H}_{t}$ contains all nodes in $S_{t}$. (Recall that a node in $S_{t}$ represents a node in $\mathcal{H}$.) For each hyperedge $e$ that contains some node in $S_{t}$, count the number of nodes that satisfy either of the following two conditions: (a) a node in a set of color $j>i$; or (b) a node in a set of color $j<i$ that has already decided to not be in the MIS (resp., GMIS). If the count is strictly smaller than $|e|-t_e$ where $t_e$ is the threshold of $e$ in $\mathcal{H}$, then we add a hyperedge $e'=e\cap S_{t}$ to $\mathcal{H}_{t}$. The threshold of $e'$ is the remaining threshold that is still available to $e$. Notice, since we use the server-client model to realize the hypergraph, in a synchronized execution, we can construct $\mathcal{H}_{t}$ in a constant number of time slots, even in the CONGEST model. In particular, in the $i$\textsuperscript{th} iteration, a server (i.e., node in $\mathcal{H}$) $u$ can first tell each adjacent client (i.e., hyperedge) $e$ about its color and whether it has decided to be in the MIS (resp., GMIS) or not. These information can be sent within one message. The client can then locally check and decide, for nodes with color $i$, whether $e'=e\cap S_{t}$ should be added to $\mathcal{H}_{t}$ or not. Next, the client can inform each adjacent server with color $i$ about whether $e'$ is constructed or not, and the remaining threshold. (However, this acknowledgment cannot contain the identities of the nodes in $e'$ due to message size constraint.)

Once $\mathcal{H}_t$ is constructed, we compute MIS (resp., GMIS) of $\mathcal{H}_{t}$. In particular, we run algorithm $\mathcal{A}$ on $\mathcal{H}_t$. Since $G_{t}$ has $O(\log{n})$ diameter, we know $\mathcal{A}$ will finish within $O(T(n))$ time if we only run it on $G_t$. However, we need to run $\mathcal{A}$ on all $G_{t_i}$ with color $i$. Nevertheless, due to property (b) in Lemma \ref{lemma-decomp-alg}, we can indeed run $\mathcal{A}$ on all $G_{t_i}$ in parallel without worrying about congestion. Hence, we can still finish executing $\mathcal{A}$ on all such $G_{t_i}$ in $O(T(n))$ time.

After running $\mathcal{A}$ for all $\mathcal{H}_t$ of all colors, the combined solution of all $\mathcal{H}_t$ will be a valid solution for the MIS (resp., GMIS) problem on hypergraph $\mathcal{H}$. We now prove the correctness of this claim. Let $M_t$ be the constructed MIS (resp., GMIS) of $\mathcal{H}_t$. Firstly, observe that any node in $M_t$ can be added to the MIS (resp., GMIS) solution of $\mathcal{H}$ without violating the threshold constraints. This is because, when constructing $\mathcal{H}_t$, for each hyperedge $e$ in $\mathcal{H}$ that contains some node in $S_t$, if $e'=e\cap S_t$ is added to $\mathcal{H}_t$, then threshold of $e$ is inherited and updated. Otherwise, if $e'$ is not added, then even if all nodes in $e$ with color $i$ decide to join the MIS (resp., GMIS), the threshold of $e$ will not be violated, as there are enough nodes in $e$ that have decided to not join the MIS (resp., GMIS), or have not decided yet. Secondly, we claim if a node $u\in S_t$ is not in $M_t$, then there exists a hyperedge $e'$ in $\mathcal{H}_t$ such that adding $u$ to $M_t$ would violate the threshold constraint of hyperedge $e$. Here, $e$ is a hyperedge in $\mathcal{H}$ and $e'=e\cap S_t$. To see this, assume adding $u$ to $M_t$ would violate the threshold constraint of $e'$ in $\mathcal{H}_t$. (We can make this assumption since $\mathcal{A}$ can correctly compute MIS (resp., GMIS) in $\mathcal{H}_t$.) Further assume there are $y_e$ nodes in $e$ that have already decided to join the MIS (resp., GMIS) when constructing $e'$. This implies the threshold associated with $e'$ is $t_e-y_e$. Moreover, adding $u$ to $M_t$ would make $t_e-y_e+1$ nodes in $e'$ decide to join the MIS (resp., GMIS). Therefore, adding $u$ to $M_t$ would make $t_e+1$ nodes in $e$ decide to join the MIS (resp., GMIS), which is a violation.

To complete the proof of the lemma, notice that we need $O(T(n))$ time for each color, and we have $O(\log{n})$ colors. Therefore, the total time complexity for computing MIS (resp., GMIS) on $\mathcal{H}$ is $O(T(n)\cdot\log{n}+\log^2{n})$.
\end{proof}

Before proceeding to the next part, we note that Lemma \ref{lemma-decomp} implies we can solve hypergraph GMIS (hence MIS as well) in $O(\log^2{n})$ time, in the distributed LOCAL model.

\begin{theorem}\label{thm-gmis-local}
A GMIS can be computed in $O(\log^2{n})$ time in the LOCAL model, w.h.p.
\end{theorem}

\begin{proof}
As stated in Lemma \ref{lemma-decomp-alg}, in $O(\log^2{n})$ time, we can decompose the input hypergraph. Then, we proceed as specified in the proof of Lemma \ref{lemma-decomp}. Notice, we are now in the LOCAL model, which means each message can be of arbitrary size. For each subgraph, since the diameter is $O(\log{n})$, by flooding information for $O(\log{n})$ time slots, all nodes in the subgraph will know everything about the constructed hypergraph, and can thus compute (identical) GMIS for the constructed hypergraph locally. This implies $T(n)=O(\log{n})$. As a result, computing GMIS for all colors takes $O(\log^2{n})$ time.
\end{proof}

\section{Computing an MIS in Linear Hypergraphs}

\subsection{The Algorithm}

In this section, we introduce a randomized distributed algorithm that solves classical MIS in linear hypergraphs, within poly-logarithmic time. As previously mentioned, it is based on a parallel algorithm originally developed by \L{}uczak and Szyma\'nska~\cite{luczak97}. Nonetheless, we have adjusted the algorithm and refined the detailed analysis accordingly, so as to greatly reduce the number of input parameters the algorithm depends upon, thus ensuring the high efficiency of this distributed variant, even in the CONGEST model.

Throughout this section, we restrict our attention to hypergraphs with diameter $O(\log{n})$, as MIS for hypergraphs with larger diameter can be computed with a poly-logarithmic time complexity overhead, due to Lemma \ref{lemma-decomp}.

Before presenting the algorithm, we introduce some relevant notations. For an $n$-node hypergraph $\mathcal{H}=(\mathcal{V},\mathcal{E})$, define $U_{i}=U_{i}(\mathcal{H})=\{e\ |\ e\in\mathcal{E},|e|=i\}$, and $u_{i}=u_{i}(\mathcal{H})=|U_{i}|$. That is, $U_{i}$ is the set of dimension $i$ hyperedges, and $u_{i}$ is the cardinality of set $U_{i}$. For a node $v\in\mathcal{V}$, define $d_i(v)=d_i(v,\mathcal{H})=|\{e\ |\ e\in U_i,v\in e\}|$. That is, $d_i(v)$ is number of dimension $i$ hyperedges that contain $v$. It is easy to see $\sum_{v\in\mathcal{V}}{d_i(v)}=iu_i$, which in turn implies the average value of $d_i(v)$ is $iu_i/n$. We say hypergraph $\mathcal{H}$ is \emph{equitable} if either $n\leq c_{eq}$ for some sufficiently large constant $c_{eq}$, or for every $i\leq\log{n}$ we have $d_i(v)\leq(iu_i/n)\cdot\log^5{n}$. (That is, for a sufficiently large hypergraph, it is ``equitable'' iff for every $i\leq\log{n}$, each node's ``dimension $i$ degree'' is not much larger than the ``average dimension $i$ degree''.)

The high level idea of the algorithm is not complicated: we initialize the independent set $\mathcal{I}$ as an empty set, and then gradually add nodes to $\mathcal{I}$; meanwhile, we also remove nodes that would violate the independence requirement if appended to $\mathcal{I}$. More specifically, the algorithm contains multiple iterations, each of which contains three parts. In the first part, we find a large equitable subhypergraph $\mathcal{H}'$ by continuously removing nodes that deviate a lot from the current average $d_i(v)$ for some $i\leq\log{n'}$, along with all the hyperedges containing any of the removed nodes. (Here, $n'$ is the number of nodes in $\mathcal{H}'$.) In the second part, we add some nodes in $\mathcal{H}'$ into a candidate set $\mathcal{W}$. The detailed rule depends on a parameter $\hat{a}$: in case $\hat{a}$ is small, we only add one special node into $\mathcal{W}$; otherwise, we add each node into $\mathcal{W}$ independently with probability $\min\{\hat{a},e^{-6}\}$. (More details regarding $\hat{a}$ will be given shortly.) Then, we remove from $\mathcal{W}$ all nodes that produce some hyperedge in $\mathcal{H}'$. The resulting set $\mathcal{I}'$ is an independent set of $\mathcal{H}'$. In the last part, we add nodes in $\mathcal{I}'$ to $\mathcal{I}$, and remove them from $\mathcal{H}$. We also remove from $\mathcal{H}$ all nodes $v\notin\mathcal{W}$ for which there exists a hyperedge $e\in\mathcal{E}$ such that $e\subseteq\{v\}\cup\mathcal{I}'$, as these nodes surely cannot be added to $\mathcal{I}$.

When implementing the above algorithm, there are some details worth clarifying.

In \L{}uczak and Szyma\'nska's original algorithm, in each iteration, in the equitable subhypergraph $\mathcal{H}'$, the aforementioned parameter $\hat{a}$ is a real value satisfying $n'/\log^8{n'}\leq\sum_{i\geq 2}{i\cdot u_i(\mathcal{H}')\cdot \hat{a}^{i-1}}\leq 2n'/\log^8{n'}$. According to this definition, to obtain $\hat{a}$, we might have to collect $\Theta(n')$ different $u_i(\mathcal{H}')$ values, resulting unacceptable time consumption. Instead, in our variant, $\hat{a}$ is defined to be a real value satisfying $n'/\log^8{n'}\leq\sum_{i=2}^{\log{n'}}{i\cdot u_i(\mathcal{H}')\cdot \hat{a}^{i-1}}\leq 2n'/\log^8{n'}$, which can be obtained much more efficiently. (We have refined the analysis to ensure correctness is still guaranteed with this updated definition.)

On the other hand, within each iteration, to calculate the value of $\hat{a}$, we need to know $n'$, as well as $u_i(\mathcal{H}')$ for each $2\leq i\leq\log{n'}$. To obtain these values, our strategy is to first elect a leader in the server-client representation of $\mathcal{H}$ (the leader can be either a server or a client), and then build a BFS tree with the root being the leader. Once the tree is built, we use aggregation to allow the root to obtain the needed values. Finally, the root broadcasts these values to all other nodes. Our procedures for accomplishing the above tasks are mostly based on the standard algorithms described in Chapter 3 and 5 of Peleg's book~\cite{peleg00}. (See Appendix \ref{appdix-obtain-para} for more details.) It is also worth noting that we cannot simply build a tree in $\mathcal{H}'$, as it might be not connected at all. Hence, during aggregation, nodes not in $\mathcal{H}'$ can simply forward values without updating them, this ensures the final results are obtained with respect to $\mathcal{H}'$.

The detailed algorithm is provided in Figure \ref{fig-protocol-linear-mis-distributed}. For simplicity, we only show the pseudocode for server nodes, and omit the pseudocode for client nodes.

\begin{figure}[t]
\hrule
\vspace{1ex}\textbf{Pseudocode executed at a node $v$ in $\mathcal{H}$:}\vspace{1ex}
\hrule
\begin{small}
\begin{algorithmic}[1]
\State $state\gets active$
\For {($l_1\gets 1$ to $\Theta(\log^{18}{n})$)}
	\If {($state\neq included$ \textbf{and} $state\neq excluded$)}
		\State $state\gets active$ \Comment If $v$ has not decided then join this iteration
	\EndIf
	\Statex \hspace{3ex}$\blacktriangleright$ \textsc{Part I}: Create equitable subhypergraph
	\State $\hat{n}\gets\texttt{CountNode}()$ \Comment Count nodes that are in $active$ state in $O(\log{n})$ time
	\For {($l_2\gets 1$ to $\Theta(\log^{2}{\hat{n}})$)}
		\State $n'\gets\texttt{CountNode}()$
		\State $\texttt{CountUi}()$ \Comment Count $u_i$ for $2\leq i\leq\log{n'}$ in $O(\log^2{n})$ time
		\If {($\texttt{CheckEq()}=true$)} \Comment Check if hypergraph is equitable in $O(\log{n})$ time
			\State \textbf{continue}
		\ElsIf {($state=active$ \textbf{and} $d_i(v)>\frac{iu_i}{n'}\cdot\log^4{n'}$ for some $i\leq\log{n'}$)}
			\State $state\gets idle$
			\State Inform adjacent client nodes about $v$ becoming $idle$ for this iteration
		\EndIf
	\EndFor
	\If {($state=idle$)} \Comment Ignore part two if $v$ is not in the equitable subhypergraph
		\State \textbf{goto} \textsc{Part III}
	\EndIf
	\Statex \hspace{3ex}$\blacktriangleright$ \textsc{Part II}: Generate an independent set
	\State Compute $\hat{a}$ such that $\frac{n'}{\log^8{n'}}\leq\sum_{i=2}^{\log{n'}}{i\cdot u_i\cdot \hat{a}^{i-1}}\leq\frac{2n'}{\log^8{n'}}$
	\State $p_0\gets\min\{\hat{a},e^{-6}\}$
	\If {($p_0\leq\frac{\log^8{n'}}{n'}$)}
		\State $v'\gets\texttt{MaxD2}()$ \Comment $\texttt{MaxD2}$ finds an active node $v'$ that maximizes $d_2(v')$ in $O(\log{n})$ time
		\If {($v=v'$)}
			\State $state\gets elected$
		\EndIf
	\ElsIf {($\texttt{Random(0,1)}\leq p_0$)} \Comment $\texttt{Random}(x,y)$ samples a random real value in $[x,y]$
		\State $state\gets elected$
	\EndIf
	\If {(there is no adjacent client $e$ s.t.\ all active servers connected to $e$ are $elected$)}
		\State $state\gets included$ \Comment $v$ decides to join the MIS
	\EndIf
	\Statex \hspace{3ex}$\blacktriangleright$ \textsc{Part III}: Update the hypergraph
	\If {($state=included$)}
		\State Inform adjacent client nodes about $v$ deciding to join the MIS
	\EndIf
	\If {(there is an adjacent client $e$ s.t.\ except $v$ all servers connected to $e$ are $included$)}
		\State $state\gets excluded$ \Comment $v$ decides to not join the MIS
	\EndIf
	\State Inform adjacent client nodes about $v$'s decision if it has decided
\EndFor
\end{algorithmic}
\end{small}
\hrule\vspace{1ex}
\caption{\textbf{Pseudocode executed at a node in $\mathcal{H}$ for computing MIS.}}\label{fig-protocol-linear-mis-distributed}
\vspace{-3ex}
\end{figure}

\subsection{The Analysis}

From the pseudocode it is easy to see the runtime of the algorithm is poly-logarithmic. Therefore, in this part, we focus on proving the correctness of the algorithm.

To begin with, we state two important observations.

\begin{fact}\label{fact-subhypergraph-mis}
Let $\mathcal{H}'$ be subhypergraph of $\mathcal{H}$, an independent set of $\mathcal{H}'$ is independent in $\mathcal{H}$ too.
\end{fact}

\begin{fact}\label{fact-linear-hypergraph-edge-num}
A linear hypergraph $\mathcal{H}$ containing $n$ nodes has at most ${n\choose 2}$ hyperedges.
\end{fact}

\begin{proof}
To see this, consider an arbitrary hyperedge $e$ in the hypergraph. If $|e|\geq 3$, then we split $e$ into two hyperedges $e_1$ and $e_2$ such that $||e_1|-|e_2||\leq 1$. If $|e_1|$ (or $|e_2|$) is of size one, then we remove $e_1$ (or $|e_2|$). (It cannot be the case that both $e_1$ and $e_2$ have size one since we require $|e|\geq 3$.) Notice, this procedure does not decrease the number of hyperedges in the hypergraph. Now, if we apply this procedure on all hyperedges recursively, we will eventually have a simple graph containing $n$ nodes. Since a simple graph with $n$ nodes has at most ${n\choose 2}$ edges, the claim is proved.
\end{proof}

The following first key technical lemma shows that within each iteration of the main algorithm, after part one, we have generated a large equitable subhypergraph containing at least half of the undecided nodes.

\begin{lemma}[Adopted from Claim 1 in \cite{luczak97}]\label{lemma-linear-mis-large-eq}
Assume at the beginning of an iteration there are $\hat{n}$ nodes in $\mathcal{H}$ that still have not decided whether to join the MIS or not. Then, after part one of this iteration, there are at least $\hat{n}/2$ nodes in $active$ state, and they induce an equitable subhypergraph.
\end{lemma}

\begin{proof}
During part one, we have a loop which contains $\Theta(\log^{2}{\hat{n}})$ inner iterations. Assume there are $\hat{n}'$ $active$ nodes at the beginning of an inner iteration. Now, if $active$ nodes do not form an equitable subhypergraph, then within this inner iteration, each $active$ node $v$ will check whether $d_i(v)>(iu_i/\hat{n}')\cdot\log^4{\hat{n}'}$ for some $i\leq\log{\hat{n}'}$. If such $d_i(v)$ exists, then $v$ will set itself as $idle$, and inform adjacent client nodes (so that these hyperedges will not be in the equitable subhypergraph).

We now argue, if at the beginning of an inner iteration, the $\hat{n}'$ $active$ nodes do not form an equitable hypergraph, and by the end of this inner iteration, the updated hypergraph is still not equitable, then by the end of this inner iteration, for some $i$ where $i\leq\log{\hat{n}'}$, the number of active dimension $i$ hyperedges decrease by at least a factor of $\log{\hat{n}'}$. To see this, notice that for such an event to happen, there must exist some node $v$ and some $i\leq\log{\hat{n}'}$ such that $d_i(v)\leq(iu_i/\hat{n}')\cdot\log^4{\hat{n}'}$ prior to this inner iteration, and $d'_i(v)>(iu'_i/\hat{n}'')\cdot\log^5{\hat{n}''}$ after this inner iteration. Notice, according to the definition, we know $d'_i(v)\leq d_i(v)$. If $u_i$ decrease by a factor less than $\log{\hat{n}'}$, then we know $d'_i(v)>(iu'_i/\hat{n}'')\cdot\log^5{\hat{n}''}>(iu_i/\hat{n}'')\cdot(1/\log{\hat{n}'})\cdot\log^5{\hat{n}''}\geq(iu_i/\hat{n}')\cdot(1/\log{\hat{n}'})\cdot\log^5{\hat{n}'}\geq d_i(v)$, a contradiction.

We then argue $\Theta(\log^{2}{\hat{n}})$ inner iterations are enough to generate an equitable subhypergraph. Assume prior to the first inner iteration, we have $x_i$ dimension $i$ hyperedges, where $1\leq i\leq\log{\hat{n}}$. Due to Fact \ref{fact-linear-hypergraph-edge-num}, we know $\sum_{i=1}^{\log{\hat{n}}}{x_i}\leq{\hat{n}\choose 2}<\hat{n}^2$, implying $x_i<\hat{n}^2$. After each inner iteration, either we have an equitable subhypergraph, or number of active dimension $i$ hyperedges is decreased by at least a factor of $\log{\hat{n}'}$ for some $i$, where $\hat{n}'$ is the number of active nodes prior to this inner iteration. Notice, once the number of active dimension $i$ hyperedges drops below one for all $i\leq\log{\hat{n}}$, the resulting subhypergraph must be equitable. On the other hand, for $x_i$ to drop below one, it is easy to see we need at most $O(\log{\hat{n}})$ inner iterations. Hence, the total number of inner iterations we need is at most $O(\log^2{\hat{n}})$.

Finally, we argue that the equitable subhypergraph generated by part one contains at least $\hat{n}/2$ nodes. To see this, notice that for arbitrary $i$, prior to an inner iteration, if there are $\hat{n}'$ $active$ nodes in total, then there are at most $\hat{n}'/\log^4{\hat{n}'}$ nodes satisfying $d_i(v)>(iu_i/\hat{n}')\cdot\log^4{\hat{n}'}$. Hence, during this iteration, we set at most $\hat{n}'/\log^3{\hat{n}'}=O(\hat{n}/\log^3{\hat{n}})$ $active$ nodes to $idle$. Since there are only $O(\log^2{\hat{n}})$ inner iterations, we know after part one, the generated equitable subhypergraph contains at least $\hat{n}/2$ nodes.
\end{proof}

In the following discussion, we focus on part two and three of each iteration. In particular, we show that if the generated equitable subhypergraph contains $c_1\leq n'\leq n$ nodes, then after part two and three, with at least some constant probability, at least $n'/\log^{17}{n}$ previously undecided nodes will make up their minds. Here, $c_1$ is a sufficiently large positive constant.

To prove the above claim, we consider three cases, depending on the value of $\hat{a}$.

The first case focuses on the scenario where $\hat{a}\leq\log^8{n'}/n'$. In such situation, there must exist a node $u$ that is contained within a lot of size two hyperedges. Thus, by letting $u$ join the MIS, the other nodes in these size two hyperedges will decide to not join the MIS.

\begin{lemma}[Adopted from Case 1 of Lemma 1 in \cite{luczak97}]\label{lemma-linear-mis-remove-nodes-case1}
Assume after part one of an iteration there are $n'$ nodes in the generated equitable subhypergraph. Further assume $\hat{a}\leq\log^8{n'}/n'$. Then, after part three of this iteration, at least $\Theta(n'/\log^{17}{n'})$ previously undecided nodes will decide whether to join the MIS or not.
\end{lemma}

\begin{proof}
First, notice that $\sum_{i\geq 3}{(iu_i\cdot\hat{a}^{i-1})}\leq\hat{a}^2\sum_{i\geq 3}{iu_i}$ when $\hat{a}\leq\log^8{n'}/n'$.

We then argue, the value of $\sum_{i\geq3}iu_i$ is at most $3\cdot{n'\choose 2}\leq(3/2)\cdot(n')^2$. To see this, we interpret $\sum_{i\geq3}iu_i$ as the sum of \emph{G3-degrees} of all the $n'$ nodes in the equitable hypergraph. Here, for a node, the G3-degree is defined as the degree of it when counting hyperedges with dimension at least three. Now, to count the sum of G3-degrees, consider the following procedure. Take an arbitrary hyperedge $e$ in the hypergraph with dimension at least three, we split $e$ into two hyperedges $e_1$ and $e_2$ such that $|e_1|=2$. If $|e_2|$ is one, then we remove $e_2$. Notice, if we apply this procedure on all hyperedges with dimension at least three recursively, we will eventually have a simple graph containing $n'$ nodes. Moreover, during the above procedure, the sum of degrees of all nodes always upper bounds the sum of G3-degrees of the original hypergraph. The only exception is that when we have a hyperedge of size one, it is removed, and this decreases the sum by one. Since such bad event can happen at most once for each of the at most ${n'\choose 2}$ hyperedges in the original hypergraph, and since for a simple graph with $n'$ nodes, the sum of all nodes' degree is at most $n'(n'-1)$, we know $\sum_{i\geq3}iu_i\leq n'(n'-1)+{n'\choose 2}= 3\cdot{n'\choose 2}$.

With the above fact, we can now conclude:

\vspace{-2ex}
\begin{align*}
\sum_{i\geq 3}{\left(iu_i\cdot\hat{a}^{i-1}\right)}\leq\hat{a}^2\cdot\sum_{i\geq 3}{iu_i}\leq\left(\log^{16}{n'}/(n')^2\right)\cdot(3/2)\cdot(n')^2=(3/2)\cdot\log^{16}{n'}
\end{align*}

Hence, we know:

\vspace{-2ex}
\begin{align*}
u_2 &= (1/2\hat{a})\cdot\left(\sum_{i=2}^{\log{n'}}{(iu_i\cdot\hat{a}^{i-1})}-\sum_{i=3}^{\log{n'}}{(iu_i\cdot\hat{a}^{i-1})}\right) \\
&\geq (1/2\hat{a})\cdot\left(n'/\log^8{n'}-(3/2)\cdot\log^{16}{n'}\right) \\
&\geq \left(n'/(2\log^8{n'})\right)\cdot\left(n'/\log^8{n'}-(3/2)\cdot\log^{16}{n'}\right) \\
&\geq (n')^2/(3\log^{16}{n'})
\end{align*}

Notice, the last inequality holds when $n'$ is sufficiently large.

Hence, assuming $v$ maximizes $d_2$, we have $d_2(v)\geq 2u_2/n'\geq 2n'/(3\log^{16}{n'})\geq n'/\log^{17}{n'}$.

Now, notice during part three, for each of the $d_2(v)$ dimension two hyperedges that contain node $v$, the other node in the hyperedge will decide to not be in the MIS (since $v$ is already in the MIS). As a result, we remove at least $d_2(v)$ nodes.
\end{proof}

The second case focuses on the scenario where $\log^8{n'}/n'\leq\hat{a}\leq e^{-6}$. This is the most involved situation. Since we have adjusted the definition of $\hat{a}$, when compared with the original proof provided in \cite{luczak97}, a refined and more careful analysis is needed to show the correctness of the following lemma. At a high-level, the proof is organized in the following way. Let $\mathcal{W}$ be the set of $elected$ nodes. We first show that with at least constant probability, there are lots of hyperedges $e$ in the equitable subhypergraph satisfying $|e\cap\mathcal{W}|=|e|-1$. Then, we prove that most of these hyperedges are vertex-disjoint and do not intersect with the hyperedges that are entirely contained in $\mathcal{W}$. Therefore, for most of the hyperedges $e$ satisfying $|e\cap\mathcal{W}|=|e|-1$, at least one node in $e$ will decide to not join the MIS.

\begin{lemma}\label{lemma-linear-mis-remove-nodes-case2}
Assume after part one of an iteration there are $n'$ nodes in the generated equitable subhypergraph. Further assume $\log^8{n'}/n'\leq\hat{a}\leq e^{-6}$. Then, after part three of this iteration, with at least constant probability, at least $\Theta(n'/\log^{8}{n'})$ previously undecided nodes will decide whether to join the MIS or not.
\end{lemma}

\begin{proof}
Before proving the lemma, we briefly recap what part two and three do. During part two, when $\log^8{n'}/n'\leq\hat{a}\leq e^{-6}$, we sample a set of nodes $\mathcal{W}$ by choosing each node independently with probability $\hat{a}$. We then construct an independent set $\mathcal{I}\subseteq\mathcal{W}$ by removing from $\mathcal{W}$ the set of nodes that constitute some hyperedge $e\subseteq E(\mathcal{H}')$. Lastly, in part three, we let a node $v$ decide to not join the MIS if it is in some hyperedge $e$ such that every node except $v$ in $e$ has already decided to join the MIS.

To prove the lemma, we rely on two key claims. The first claim shows that with at least some constant probability there are lots of hyperedges $e$ in the equitable subhypergraph satisfying $|e\cap\mathcal{W}|=|e|-1$. The second claim shows that most of these hyperedges are vertex-disjoint and do not intersect the hyperedges that are entirely contained in $\mathcal{W}$.

Let $X$ be a random variable denoting the number of hyperedges in the equitable subhypergraph such that for each such hyperedge all but one of its nodes are in $\mathcal{W}$. The first claim, as mentioned previously, estimates the value of $X$.

\begin{claim}
With at least constant probability, $X=\Theta(n'/\log^8{n'})$.
\end{claim}

\begin{proof}
For a hyperedge $e$ in the generated equitable subhypergraph $\mathcal{H}'=(\mathcal{V}',\mathcal{E}')$, define $X_e$ to be an indicator random variable taking value one iff $|e\cap\mathcal{W}|=|e|-1$.

It is easy to see:

\vspace{-2ex}
\begin{align*}
\mathbb{E}(X) & =\sum_{e\in\mathcal{E}'}{\mathbb{E}(X_e)}=\sum_{e\in\mathcal{E}'}{\left(|e|\cdot(1-\hat{a})\cdot\hat{a}^{|e|-1}\right)} \\
& =(1-\hat{a})\cdot\sum_{e\in\mathcal{E}'}{\left(|e|\cdot\hat{a}^{|e|-1}\right)} = (1-\hat{a})\cdot\sum_{i\geq 2}{\left(iu_i\cdot\hat{a}^{i-1}\right)} \\
& =\Theta(1)\cdot\left(\sum_{i=2}^{\log{n'}}{\left(iu_i\cdot\hat{a}^{i-1}\right)}+\sum_{i>\log{n'}}{\left(iu_i\cdot\hat{a}^{i-1}\right)}\right) \\
& =\Theta(1)\cdot\left(\Theta\left(\frac{n'}{\log^8{n'}}\right)+\sum_{i>\log{n'}}{\left(iu_i\cdot\hat{a}^{i-1}\right)}\right)
\end{align*}

Notice that:

\vspace{-2ex}
\begin{align*}
\sum_{i>\log{n'}}{\left(iu_i\cdot\hat{a}^{i-1}\right)} & \leq\sum_{i>\log{n'}}{\left(iu_i\cdot \left(e^{-6}\right)^{(i-1)}\right)} \leq \sum_{i>\log{n'}}{\left(iu_i\cdot e^{-6\log{n'}}\right)} \\
& =e^{-6\log{n'}}\cdot\sum_{i>\log{n'}}{iu_i} \leq e^{-6\log{n'}}\cdot n'\cdot{n'\choose 2} \\
& =O\left(\left(n'\right)^{-3}\right)
\end{align*}

As a result, we know $\mathbb{E}(X)=\Theta(n'/\log^8{n'})$.

To show $X$ is not likely to deviate much from its expectation, we will use the Chebyshev's inequality~\cite{mitzenmacher17}, which in turn requires us to calculate the variance of $X$.

By the definition of variance, we know:

\vspace{-2ex}
\begin{align*}
\mathrm{Var}(X) =\mathrm{Var}\left(\sum_{e\in\mathcal{E'}}{X_e}\right) = \quad \sum_{e\in\mathcal{E'}}{\mathrm{Var}(X_e)} \quad + \quad \sum_{e,e'\in\mathcal{E'};e\cap e'\neq\emptyset}{\mathrm{Cov}(X_e,X_{e'})}
\end{align*}

Since $\mathrm{Var}(X_e)=\mathbb{E}(X_e^2)-(\mathbb{E}(X_e))^2\leq\mathbb{E}(X_e^2)=\mathbb{E}(X_e)$, we know $\sum_{e\in\mathcal{E'}}{\mathrm{Var}(X_e)}\leq\mathbb{E}(X)$.

On the other hand:

\vspace{-2ex}
\begin{align*}
\sum_{e,e'\in\mathcal{E'};e\cap e'\neq\emptyset}{\mathrm{Cov}(X_{e},X_{e'})} \quad = & \quad \sum_{e,e'\in\mathcal{E'};e\cap e'\neq\emptyset}{\left(\mathbb{E}(X_{e} X_{e'})-\mathbb{E}(X_{e}) \mathbb{E}(X_{e'})\right)} \\
\leq & \quad \sum_{e,e'\in\mathcal{E'};e\cap e'\neq\emptyset}{\mathbb{E}(X_{e}X_{e'})}
\end{align*}

Since $\mathcal{H'}$ is a linear hypergraph, we can conclude:

\vspace{-2ex}
\begin{align*}
\phantom{=} & \quad \sum_{e,e'\in\mathcal{E'};e\cap e'\neq\emptyset}{\mathbb{E}(X_{e}X_{e'})}\\
= & \quad \sum_{e,e'\in\mathcal{E'};e\cap e'\neq\emptyset}{\left( (|e|-1)(|e'|-1)(1-\hat{a})^2\cdot\hat{a}^{|e|+|e'|-3}\right)} +\\
\phantom{+} & \quad \sum_{e,e'\in\mathcal{E'};e\cap e'\neq\emptyset}{\left((1-\hat{a})\cdot\hat{a}^{|e|+|e'|-2}\right)}\\
\leq & \quad 2\cdot\sum_{e,e'\in\mathcal{E'};e\cap e'\neq\emptyset}{(|e|\cdot|e'|\cdot\hat{a}^{|e|+|e'|-3})}\\
= & \quad 2\cdot\sum_{e\in\mathcal{E'}}{\left( |e|\cdot\hat{a}^{|e|-1}\cdot\sum_{e'\in\mathcal{E'};e\cap e'\neq\emptyset}{(|e'|\cdot\hat{a}^{|e'|-2})} \right)}\\
\end{align*}

Notice, in the above, the first equality holds since $X_e X_{e'}=1$ iff $X_e=X_{e'}=1$, which can only happen in one of the two following cases: (a) $e\cap e'=\{u\}$, all nodes in $(e\cup e')-\{u\}$ is marked and $u$ is not marked; or (b) $e\cap e'=\{u\}$, one node in $e-\{u\}$ is not marked, one node in $e'-\{u\}$ is not marked, and all other nodes in $e\cup e'$ are marked.

Our next step is to obtain an upper bound for $\sum_{e,e'\in\mathcal{E'}; e\cap e'\neq\emptyset}{\mathrm{Cov}(X_{e},X_{e'})}$, by bounding $\sum_{e\in\mathcal{E'}}{(|e|\cdot\hat{a}^{|e|-1}\cdot\sum_{e'\in\mathcal{E'};e\cap e'\neq\emptyset}{(|e'|\cdot\hat{a}^{|e'|-2})})}$. Define $\Delta_i=\max_{v\in\mathcal{V'}}{d_i(v)}$. Since $\mathcal{H'}$ is equitable, we know $\Delta_i\leq(iu_i/n')\cdot\log^5{n'}$ for $2\leq i\leq\log{n'}$. The analysis in Figure \ref{fig-eqnarray-1} shows $\sum_{e,e'\in\mathcal{E'}; e\cap e'\neq\emptyset}{\mathrm{Cov}(X_{e},X_{e'})}\in O((\mathbb{E}(X))^2/(\log{n'}))$, and some explanations are needed:

\begin{itemize}
	\item To see inequality (\ref{eqn-mis-case2-leq-2}), notice $\sum_{e\in\mathcal{E'}}\sum_{e'\in\mathcal{E'};e\cap e'\neq\emptyset}{(|e'|\cdot\hat{a}^{|e'|-2})}$ is equal to the sum of $\sum_{e\in\mathcal{E'};|e|\leq\log{n'}}\sum_{e'\in\mathcal{E'};e\cap e'\neq\emptyset}{(|e'|\cdot\hat{a}^{|e'|-2})}$ and $\sum_{e\in\mathcal{E'};|e|>\log{n'}}\sum_{e'\in\mathcal{E'};e\cap e'\neq\emptyset}{(|e'|\cdot\hat{a}^{|e'|-2})}$. Let us focus on $\sum_{e\in\mathcal{E'};|e|\leq\log{n'}}\sum_{e'\in\mathcal{E'};e\cap e'\neq\emptyset}{(|e'|\cdot\hat{a}^{|e'|-2})}$. Fix a hyperedge $e\in\mathcal{E'}$ such that $|e|\leq\log{n'}$, we now bound $\sum_{e'\in\mathcal{E'};e\cap e'\neq\emptyset}{(|e'|\cdot\hat{a}^{|e'|-2})}$. For each node $v\in e$, for arbitrary $i$, by the definition of $\Delta_i$, we know $d_i(v)\leq\Delta_i$. That is, for each node $v\in e$, for arbitrary $i$, node $v$ is contained within at most $\Delta_i$ dimension $i$ hyperedges; or, put another way, there are at most $\Delta_i-1$ dimension $i$ hyperedges in $\mathcal{E'}$ that intersect with $e$ on node $v$. Since $|e|\leq\log{n'}$, we know $\sum_{e'\in\mathcal{E'};e\cap e'\neq\emptyset}{(|e'|\cdot\hat{a}^{|e'|-2})}\leq\log{n'}\cdot\sum_{i\geq2}{(i\Delta_i \hat{a}^{i-2})}$. Similarly, when $|e|\geq\log{n'}$, we know $\sum_{e'\in\mathcal{E'};e\cap e'\neq\emptyset}{(|e'|\cdot\hat{a}^{|e'|-2})}\leq n'\cdot\sum_{i\geq2}{(i\Delta_i \hat{a}^{i-2})}$, as the maximum size for any hyperedge is bounded by $n'$.
	
	\item To see inequality (\ref{eqn-mis-case2-leq-4}), notice $\sum_{i>\log{n'}}{(i\Delta_i\hat{a}^{i-2})}\leq n'\cdot(n'\cdot{n'\choose 2}\cdot(e^{-6})^{\log{n'}-1})=O(1/(n')^2)$. In the meantime, $\sum_{i=2}^{\log{n'}}{(i\Delta_i\hat{a}^{i-2})}\leq\sum_{i=2}^{\log{n'}}{((i^2u_i/n')\cdot\log^5{n'}\cdot\hat{a}^{i-2})}\leq(\log^6{n'}/(\hat{a}n'))\cdot\sum_{i=2}^{\log{n'}}{(i\cdot u_i\cdot\hat{a}^{i-1})}\leq (1/\log^2{n'})\cdot(2n'/\log^8{n'})=O(n'/\log^{10}{n'})$.
	
	\item To see inequality (\ref{eqn-mis-case2-leq-5}), notice $\sum_{|e|>\log{n'}}{(|e|\cdot\hat{a}^{|e|-1}\cdot(n'(O({n'}/{\log^{10}{n'}})+O({1}/{(n')^2}))))}=\sum_{|e|>\log{n'}}{(|e|\cdot\hat{a}^{|e|-1}\cdot(n'\cdot O({n'}/{\log^{10}{n'}})))}=O((n')^2/\log^{10}{n'})\cdot\sum_{|e|>\log{n'}}{(|e|\cdot\hat{a}^{|e|-1})}$. Moreover, $\sum_{|e|>\log{n'}}{(|e|\cdot\hat{a}^{|e|-1})}\leq{n'\choose 2}\cdot(n'\cdot(e^{-6})^{\log{n'}})\leq 1/(n')^3$. As a result, we know $\sum_{e\in\mathcal{E'};|e|>\log{n'}}{(|e|\cdot\hat{a}^{|e|-1}\cdot(n'(O({n'}/{\log^{10}{n'}})+O({1}/{(n')^2}))))}\leq O(1/(n'\cdot\log^{10}{n'}))$.
	
	\item To see inequality (\ref{eqn-mis-case2-leq-6}), notice that $\sum_{|e|\leq\log{n'}}{(|e|\cdot\hat{a}^{|e|-1}\cdot O({1}/{(n')^2})))}=O(\log{n'}/(n')^2)\cdot\sum_{|e|\leq\log{n'}}{(|e|\cdot\hat{a}^{|e|-1})}$. In the meantime, $\sum_{|e|\leq\log{n'}}{(|e|\cdot\hat{a}^{|e|-1})}=\sum_{i=2}^{\log{n'}}{(i\cdot u_i\cdot\hat{a}^{i-1})} \leq 2n'/\log^8{n'}$. Therefore, $\sum_{e\in\mathcal{E'};|e|\leq\log{n'}}{(|e|\cdot\hat{a}^{|e|-1}\cdot O({1}/{(n')^2})))}\leq O(1/(n'\cdot\log^{7}{n'}))$.
\end{itemize}

\begin{figure}[t!]
\begin{small}
\begin{align*}
\phantom{\leq} & \quad \sum_{e,e'\in\mathcal{E'}; e\cap e'\neq\emptyset}{\mathrm{Cov}(X_{e},X_{e'})}\\
\leq & \quad 2\cdot\sum_{e\in\mathcal{E'}}{\left( |e|\cdot\hat{a}^{|e|-1}\cdot\sum_{e'\in\mathcal{E'};e\cap e'\neq\emptyset}{(|e'|\cdot\hat{a}^{|e'|-2})} \right)}\\
\leq & \quad 2\cdot\sum_{e\in\mathcal{E'};|e|\leq\log{n'}}{\left(|e|\cdot\hat{a}^{|e|-1}\cdot\left(\log{n'}\cdot\sum_{i\geq2}{(i\Delta_i \hat{a}^{i-2})}\right)\right)} + \numberthis\label{eqn-mis-case2-leq-2} \\
\phantom{+} & \quad 2\cdot\sum_{e\in\mathcal{E'};|e|>\log{n'}}{\left(|e|\cdot\hat{a}^{|e|-1}\cdot\left(n'\cdot\sum_{i\geq 2}{(i\Delta_i \hat{a}^{i-2})}\right)\right)}\\
= & \quad 2\cdot\sum_{e\in\mathcal{E'};|e|\leq\log{n'}}{\left(|e|\cdot\hat{a}^{|e|-1}\cdot\left(\log{n'}\left(\sum_{i=2}^{\log{n'}}{(i\Delta_i \hat{a}^{i-2})}+\sum_{i>\log{n'}}{(i\Delta_i \hat{a}^{i-2})}\right)\right)\right)} +\\
\phantom{+} & \quad 2\cdot\sum_{e\in\mathcal{E'};|e|>\log{n'}}{\left(|e|\cdot\hat{a}^{|e|-1}\cdot\left(n'\left(\sum_{i=2}^{\log{n'}}{(i\Delta_i \hat{a}^{i-2})}+\sum_{i>\log{n'}}{(i\Delta_i \hat{a}^{i-2})}\right)\right)\right)}\\
\leq & \quad 2\cdot\sum_{e\in\mathcal{E'};|e|\leq\log{n'}}{\left(|e|\cdot\hat{a}^{|e|-1}\cdot\left(\log{n'}\left(\sum_{i=2}^{\log{n'}}{(i\Delta_i \hat{a}^{i-2})}+O\left(\frac{1}{(n')^2}\right)\right)\right)\right)} + \numberthis\label{eqn-mis-case2-leq-4}\\
\phantom{+} & \quad 2\cdot\sum_{e\in\mathcal{E'};|e|>\log{n'}}{\left(|e|\cdot\hat{a}^{|e|-1}\cdot\left(n'\left(O\left(\frac{n'}{\log^{10}{n'}}\right)+O\left(\frac{1}{(n')^2}\right)\right)\right)\right)}\\
\leq & \quad 2\cdot\sum_{e\in\mathcal{E'};|e|\leq\log{n'}}{\left(|e|\cdot\hat{a}^{|e|-1}\cdot\left(\log{n'}\left(\sum_{i=2}^{\log{n'}}{(i\Delta_i \hat{a}^{i-2})}+O\left(\frac{1}{(n')^2}\right)\right)\right)\right)} + \numberthis\label{eqn-mis-case2-leq-5}\\
\phantom{+} & \quad O\left(\frac{1}{n'\cdot\log^{10}{n'}}\right)\\
\leq & \quad 2\cdot\sum_{e\in\mathcal{E'};|e|\leq\log{n'}}{\left(|e|\cdot\hat{a}^{|e|-1}\cdot\log{n'}\cdot\sum_{i=2}^{\log{n'}}{(i\Delta_i \hat{a}^{i-2})}\right)} + O\left(\frac{1}{n'\log^7{n'}}\right) + \numberthis\label{eqn-mis-case2-leq-6}\\
\phantom{+} & \quad O\left(\frac{1}{n'\cdot\log^{10}{n'}}\right)\\
\leq & \quad \frac{2\log^7{n'}}{\hat{a}\cdot n'}\cdot\sum_{e\in\mathcal{E'};|e|\leq\log{n'}}{\left( |e|\cdot\hat{a}^{|e|-1}\cdot\left(\sum_{i=2}^{\log{n'}}{i\cdot u_i\cdot\hat{a}^{i-1}}\right) \right)}+O(1)\\
\leq & \quad \frac{2\log^7{n'}}{\hat{a}\cdot n'\cdot(1-\hat{a})}\cdot\mathbb{E}(X)\cdot\left(\sum_{i=2}^{\log{n'}}{i\cdot u_i\cdot\hat{a}^{i-1}}\right)+O(1)\\
\leq & \quad \frac{2\log^7{n'}}{\hat{a}\cdot n'\cdot(1-\hat{a})}\cdot\mathbb{E}(X)\cdot\frac{2n'}{\log^8{n'}}+O(1)\\
\leq & \quad O\left(\frac{1}{\log{n'}}\right)\cdot\mathbb{E}(X)\cdot\mathbb{E}(X)+O(1)\\
= & \quad O\left(\frac{(\mathbb{E}(X))^2}{\log{n'}}\right)\\
\end{align*}
\end{small}
\vspace{-6ex}
\caption{\textbf{Bounding $\sum_{e,e'\in\mathcal{E'}; e\cap e'\neq\emptyset}{\mathrm{Cov}(X_{e},X_{e'})}$ to $O((\mathbb{E}(X))^2/(\log{n'}))$.}}\label{fig-eqnarray-1}
\end{figure}

At this point, we can conclude $\mathrm{Var}(X)\leq\mathbb{E}(X)+O((\mathbb{E}(X))^2/\log{n'})=O((\mathbb{E}(X))^2/\log{n'})$. Apply the Chebyshev's inequality, and our claim follows.
\end{proof}

We then prove our second claim, which states there are only few pairs of intersecting hyperedges in $\mathcal{H'}$ which share a large number of nodes with $\mathcal{W}$. More precisely:

\begin{claim}
With at least constant probability, equitable hypergraph $\mathcal{H'}$ contains at most $O(n'/\log^9{n'})$ pairs of hyperedges $e,e'$ for which $e\cap e'\neq\emptyset$, $|e\cap\mathcal{W}|\geq|e|-1$, and $e'\backslash e\subseteq\mathcal{W}$.
\end{claim}

\begin{proof}
Let $Y$ denote the number of such pairs of hyperedges, since $\mathcal{H'}$ is a linear hypergraph, by the analysis shown in Figure \ref{fig-eqnarray-2}, we know $\mathbb{E}(Y)$ is at most $O(n'/(\log^{10}{n'}))$. As a result, the claim follows by Markov's inequality.
\end{proof}

\begin{figure}[t!]
\begin{align*}
\mathbb{E}(Y) \quad \leq & \quad \sum_{e\in\mathcal{E'}}{\left( |e|\cdot\hat{a}^{|e|-1}\cdot\sum_{e'\in\mathcal{E'};e\cap e'\neq\emptyset}{\hat{a}^{|e'|-1}} \right)}\\
= & \quad \sum_{i=2}^{\log{n'}}{\left( i\cdot u_i\cdot\hat{a}^{i-1}\cdot\sum_{e'\in\mathcal{E'};e\cap e'\neq\emptyset}{\hat{a}^{|e'|-1}} \right)} + \sum_{i>\log{n'}}{\left( i\cdot u_i\cdot\hat{a}^{i-1}\cdot\sum_{e'\in\mathcal{E'};e\cap e'\neq\emptyset}{\hat{a}^{|e'|-1}} \right)}\\
\leq & \quad \sum_{i=2}^{\log{n'}}{\left( i\cdot u_i\cdot\hat{a}^{i-1}\cdot\log{n'}\cdot\sum_{j\geq 2}{(\Delta_j\cdot\hat{a}^{j-1})} \right)} + \sum_{i>\log{n'}}{\left( i\cdot u_i\cdot\hat{a}^{i-1}\cdot n'\cdot\sum_{j\geq 2}{(\Delta_j\cdot\hat{a}^{j-1})} \right)}\\
\leq & \quad \sum_{i=2}^{\log{n'}}{\left( i\cdot u_i\cdot\hat{a}^{i-1}\cdot\log{n'}\cdot\left(\sum_{j=2}^{\log{n'}}{(\Delta_j\cdot\hat{a}^{j-1})}+\sum_{j>\log{n'}}{(\Delta_j\cdot\hat{a}^{j-1})}\right) \right)} + \\
\phantom{=} & \quad \sum_{i>\log{n'}}{\left( i\cdot u_i\cdot\hat{a}^{i-1}\cdot n'\cdot\left(\sum_{j=2}^{\log{n'}}{(\Delta_j\cdot\hat{a}^{j-1})}+\sum_{j>\log{n'}}{(\Delta_j\cdot\hat{a}^{j-1})}\right) \right)}\\
\leq & \quad \sum_{i=2}^{\log{n'}}{\left( i\cdot u_i\cdot\hat{a}^{i-1}\cdot\log{n'}\cdot\left(O\left(\frac{1}{\log^3{n'}}\right)+O\left(\frac{1}{(n')^3}\right)\right) \right)} + \\
\phantom{=} & \quad \sum_{i>\log{n'}}{\left( i\cdot u_i\cdot\hat{a}^{i-1}\cdot n'\cdot\left(O\left(\frac{1}{\log^3{n'}}\right)+O\left(\frac{1}{(n')^3}\right)\right) \right)}\\
\leq & \quad \sum_{i=2}^{\log{n'}}{\left( i\cdot u_i\cdot\hat{a}^{i-1}\cdot\log{n'}\cdot O\left(\frac{1}{\log^3{n'}}\right) \right)} + \sum_{i>\log{n'}}{\left( i\cdot u_i\cdot\hat{a}^{i-1}\cdot n'\cdot O\left(\frac{1}{\log^3{n'}}\right) \right)}\\
\leq & \quad O\left(\frac{n'}{\log^{10}{n'}}\right) + O\left(\frac{1}{(n')^2\cdot\log^3{n'}}\right) = O\left(\frac{n'}{\log^{10}{n'}}\right)\\
\end{align*}
\vspace{-6ex}
\caption{\textbf{Bounding $\mathbb{E}(Y)$ to $O(n'/(\log^{10}{n'}))$.}}\label{fig-eqnarray-2}
\end{figure}

We now prove the lemma. The above two claims show that in each iteration, with at least constant probability, in $\mathcal{H'}$ there exists a set $\tilde{\mathcal{E}}'$ of hyperedges of cardinality $\Theta(n'/\log^8{n'})$ such that: (a) for $e\in\tilde{\mathcal{E}}'$ we have $e\backslash\mathcal{W}=\{v_e\}$; (b) no $e$ from $\tilde{\mathcal{E}}'$ share a node with a hyperedge of $\mathcal{H'}$ entirely contained in $\mathcal{W}$; and (c) for $e, e'\in\tilde{\mathcal{E}}'$, $v_{e}\neq v_{e'}$ whenever $e\neq e'$.

Now, notice that (a) and (b) imply that after part three of the iteration, for each hyperedge in $\tilde{\mathcal{E}}'$, at least one node has decided to not be in the MIS. Moreover, condition (c) guarantees that these nodes are different. Therefore, we have proved the lemma.
\end{proof}

The last case focuses on the scenario where $\hat{a}\geq e^{-6}$. In such situation, $\Theta(n')$ undecided nodes will be marked (i.e., $elected$), yet entirely marked hyperedges will contain at most $O(n'/\log^8{n'})$ nodes. As a result, we know $\Theta(n')$ nodes will decide to join the MIS.

\begin{lemma}[Adopted from Case 3 of Lemma 1 in \cite{luczak97}]\label{lemma-linear-mis-remove-nodes-case3}
Assume after part one of an iteration there are $n'$ nodes in the generated equitable subhypergraph. Further assume $\hat{a}\geq e^{-6}$. Then, after part three of this iteration, with at least constant probability, at least $\Theta(n')$ previously undecided nodes will decide whether to join the MIS or not.
\end{lemma}

\begin{proof}
Since $\hat{a}\geq e^{-6}$, we know each node in the generated equitable hypergraph $\mathcal{H'}$ will be selected with probability $e^{-6}$. By a Chernoff bound~\cite{mitzenmacher17}, we know w.h.p.\ w.r.t.\ $n'$, $\Theta(n')$ nodes will be selected into $\mathcal{W}$. I.e., $|\mathcal{W}|=\Theta(n')$ with at least constant probability.

On the other hand, in expectation, the number of nodes that belong to some hyperedges that are entirely contained  in $\mathcal{W}$ is upper bounded by $\sum_{i\geq 2}{i\cdot u_i\cdot p_0^i}=\sum_{i=2}^{\log{n'}}{i\cdot u_i\cdot p_0^i}+\sum_{i>\log{n'}}{i\cdot u_i\cdot p_0^i}\leq e^{-6}\cdot\sum_{i=2}^{\log{n'}}{i\cdot u_i\cdot \hat{a}^{i-1}}+O(1/(n')^2)=O(n'/\log^8{n'})$.

Therefore, by a Markov's inequality, we know with at least some constant probability, after part two of the iteration, we can find an independent set of size $\Theta(n')$. Moreover, these nodes will decide to join the MIS by the end of this iteration.
\end{proof}

Combine the above four lemmas, we can conclude if prior to an iteration there are $n'$ undecided nodes, then after this iteration, with at least constant probability, at least $\Omega(n'/\log^{17}{n'})$ nodes will decide, provided $n'$ is sufficiently large. Since $n'\leq n$, this means after $O(\log^{18}{n})$ iterations, the number of undecided nodes will be reduced to some sufficiently large constant $c_1$, w.h.p.

Once the number of undecided nodes is reduced to $c_1$, during part two of an iteration, one of the two following situations will happen: (a) $p_0\leq\log^8{n'}/n'$, in which case only one node is selected into $\mathcal{W}$; or (b) $e^{-6}\geq p_0>\log^8{n'}/n'$, in which case each of the $n'$ nodes is selected with probability $p_0$. In the first case, the single selected node will decide to join the MIS after this iteration. In the second case, since $n'\leq c_1$ is a constant, we know with at least some constant probability only one of the $n'$ nodes will be selected into $\mathcal{W}$, and will decide to join the MIS after this iteration. Either way, we know after each iteration, with at least some constant probability, one node will decide to join the MIS.

At this point, we can conclude that after at most some poly-logarithmic (w.r.t.\ $n$) iterations, all nodes in $\mathcal{H}$ will decide, w.h.p. Moreover, it is easy to see that the result is indeed an MIS of $\mathcal{H}$. Combine these with Lemma \ref{lemma-decomp}, we immediately have the following theorem.

\begin{theorem}\label{thm-linear-mis}
In the CONGEST model, there exists a distributed algorithm that can solve the MIS problem for linear hypergraphs within poly-logarithmic time, w.h.p.
\end{theorem}

\section{Computing a GMIS in Constant Dimension Linear Hypergraphs}

\subsection{The Algorithm}

To compute a generalized maximal independent set (GMIS) for a linear hypergraph, we take a similar approach as in the MIS case: first decompose the input hypergraph; then run the core GMIS algorithm within the generated subhypergraphs in parallel; finally, combine these partial solutions to obtain a complete GMIS for the original input hypergraph.

We once again restrict our attention to hypergraphs with $O(\log{n})$ diameter (see Lemma \ref{lemma-decomp}). We also restrict the dimension of the input hypergraph to be some constant $d$. Towards the end of the paper, we will discuss why this limitation is posed.

The core algorithm for computing GMIS in low-diameter hypergraphs is a non-trivial generalization of our previous hypergraph MIS algorithm. Before presenting more details, we introduce some updated notations.

The first one is \emph{strict subhypergraph}. Consider a hypergraph $\mathcal{H}=(\mathcal{V},\mathcal{E})$, for a subset $\mathcal{W}$ of $\mathcal{V}$, the induced strict subhypergraph $\mathcal{H}_{\mathcal{W}}=(\mathcal{W},\mathcal{E}_{\mathcal{W}})$ is defined as: for each $v\in\mathcal{V}\backslash\mathcal{W}$, delete $v$ from each hyperedge $e\in\mathcal{E}$, the threshold attached with $e$ remains unchanged; then, for each remaining hyperedge $e'=e\cap\mathcal{W}$, delete $e'$ if $|e'|\leq t_{e'}$. (Notice $t_{e'}=t_e$.)

The reason for defining strict subhypergraph is to maintain a property that is critical to the correctness of our algorithm. More specifically, consider a hypergraph $\mathcal{H}$ and one of its subhypergraph $\mathcal{H}'$. When dealing with the MIS problem, an independent set of $\mathcal{H}'$ is also independent in $\mathcal{H}$. (See Fact \ref{fact-subhypergraph-mis}.) However, for generalized independent sets, this is no longer the case. By contrast, if $\mathcal{H}''$ is a strict subhypergraph of $\mathcal{H}$, then a generalized independent set of $\mathcal{H}''$ is also a generalized independent set of $\mathcal{H}$. That is, we have:

\begin{fact}\label{fact-strict-subhypergraph-gmis}
If $\mathcal{H}_{\mathcal{W}}$ is a strict subhypergraph of $\mathcal{H}$, then a generalized independent set of $\mathcal{H}_{\mathcal{W}}$ is also a generalized independent set of $\mathcal{H}$.
\end{fact}

\begin{proof}
We prove the claim by contradiction. Assume $I_{\mathcal{W}}$ is a generalized independent set of $\mathcal{H}_{\mathcal{W}}=(\mathcal{W},\mathcal{E}_{\mathcal{W}})$, but not a generalized independent set of $\mathcal{H}=(\mathcal{V},\mathcal{E})$. Then there must exist a hyperedge $e\in\mathcal{E}$ such that $|e\cap I_{\mathcal{W}}|>t_e$. Notice, all the nodes in $I_{\mathcal{W}}$ are also in $\mathcal{W}$. Therefore, $t_e<|e\cap I_{\mathcal{W}}|\leq|e\cap\mathcal{W}|$. As a result, during the construction of $\mathcal{H}_{\mathcal{W}}$, a hyperedge $e'=e\cap\mathcal{W}$ is added to $\mathcal{H}_\mathcal{W}$, with threshold $t_{e'}=t_e$. However, recall that $|e'\cap I_{\mathcal{W}}|=|e\cap\mathcal{W}\cap I_{\mathcal{W}}|=|e\cap I_{\mathcal{W}}|>t_e=t_{e'}$, which contradicts the assumption that $I_{\mathcal{W}}$ is a generalized independent set of $\mathcal{H}_{\mathcal{W}}$. Thus our claim is proved.
\end{proof}

On the other hand, we have also significantly adjusted the definition of equitable hypergraph. In this section, for a hypergraph $\mathcal{H}=(\mathcal{V},\mathcal{E})$, define $U_t(\mathcal{H})=U_t$ to be the set of hyperedges with threshold $t$, and $u_t(\mathcal{H})=u_t$ to be the cardinality of $U_t$. Define $d_t(v,\mathcal{H})=d_t(v)$ to be the number of hyperedges that contain $v$ and have threshold $t$. Now, for a hypergraph $\mathcal{H}$, we say it is \emph{equitable} if it contains less than $c_{eq}$ nodes (where $c_{eq}$ is a sufficiently large constant), or for each node $v$ and each $i$ where $1\leq i\leq d-1$ we have $d_i(v)\leq(i u_i/n)\cdot\log^5{n}$.

We now describe the algorithm for computing GMIS in a low-diameter constant dimension linear hypergraph $\mathcal{H}$. The algorithm contains multiple iterations, each of which has three parts. In the first part, we try to find a large equitable strict subhypergraph $\mathcal{H'}$. In the second part, we add some nodes in $\mathcal{H}'$ into a candidate set $\mathcal{W}$. The detailed rule depends on a parameter $\hat{a}$ satisfying $n'/\log^8{n'}\leq\sum_{i=1}^{d-1}{u_i\cdot \hat{a}^{i}}\leq 2n'/\log^8{n'}$. (Notice, this definition of $\hat{a}$ is quite different from the one we used in our previous linear hypergraph MIS algorithm.) In case $\hat{a}$ is small, we add one node $v$ which maximizes $d_1(v)$ into $\mathcal{W}$; otherwise, we independently add each node into $\mathcal{W}$ with probability $\min\{\hat{a},e^{-6}\}$. Then, we remove from $\mathcal{W}$ all nodes that would violate some hyperedge's threshold constraint in $\mathcal{H}'$. The resulting set $\mathcal{I}'$ is a generalized independent set of $\mathcal{H}'$. In the last part, we add nodes from $\mathcal{I}'$ to $\mathcal{I}$, and remove them from $\mathcal{H}$. We also remove from $\mathcal{H}$ all nodes $v\notin\mathcal{W}$ for which there exists a hyperedge $e\in\mathcal{E}$ such that $|(\{v\}\cup\mathcal{I}')\cap e|>t_e$, as these nodes cannot be added into $\mathcal{I}$.

The detailed algorithm is shown in Figure \ref{fig-protocol-linear-gmis-distributed}. For simplicity, we again only include the pseudocode for server nodes. Moreover, to obtain the values of $n'$ and $u_i$ (and some other parameters), we reuse the aggregation procedures described in earlier sections.

\begin{figure}[t]
\hrule
\vspace{1ex}\textbf{Pseudocode executed at a node $v$ in $\mathcal{H}$:}\vspace{1ex}
\hrule
\begin{small}
\begin{algorithmic}[1]
\State $state\gets active$
\For {($l_1\gets 1$ to $\Theta(\log^{18}{n})$)}
	\If {($state\neq included$ and $state\neq excluded$)}
		\State $state\gets active$ \Comment If $v$ has not decided then join this iteration
	\EndIf
	\Statex \hspace{3ex}$\blacktriangleright$ \textsc{Part I}: Create equitable strict subhypergraph
	\State $\hat{n}\gets\texttt{CountNode}()$ \Comment Count nodes that are still in $active$ state in $O(\log{n})$ time
	\For {($l_2\gets 1$ to $\Theta(d\cdot\log{\hat{n}})$)}
		\State $n'\gets\texttt{CountNode}()$
		\State $\texttt{CountUi}()$ \Comment Count $u_i$ for $1\leq i\leq d-1$, takes $O(d\cdot\log{n})$ time
		\If {($\texttt{CheckEq()}=true$)} \Comment Checks if hypergraph is equitable in $O(\log{n})$ time
			\State \textbf{continue}
		\ElsIf {($state=active$ \textbf{and} $d_i(v)>\frac{iu_i}{n'}\cdot\log^4{n'}$ for some $1\leq i\leq d-1$)}
			\State $state\gets idle$
			\State Inform adjacent client nodes about $v$ becoming $idle$ for this iteration
		\EndIf
	\EndFor
	\If {($state=idle$)} \Comment Ignore part two if $v$ is not in the equitable strict subhypergraph
		\State \textbf{goto} \textsc{Part III}
	\EndIf
	\Statex \hspace{3ex}$\blacktriangleright$ \textsc{Part II}: Generate a generalized independent set
	\State Compute $\hat{a}$ such that $\frac{n'}{\log^8{n'}}\leq\sum_{i=1}^{d-1}{u_i\cdot \hat{a}^{i}}\leq\frac{2n'}{\log^8{n'}}$
	\State $p_0\gets\min\{\hat{a},e^{-6}\}$
	\If {($p_0\leq\frac{\log^8{n'}}{n'}$)}
		\State $v'\gets\texttt{MaxD1}()$ \Comment $\texttt{MaxD1}$ finds an active node $v'$ that maximizes $d_1(v')$ in $O(\log{n})$ time
		\If {($v=v'$)} $state\gets elected$ \EndIf
	\ElsIf {($\texttt{Random(0,1)}\leq p_0$)} $state\gets elected$
	\EndIf
	\If {(there is no adjacent client $e$ s.t.\ $elected$ servers connected to $e$ exceed $t_e$)}
		\State $state\gets included$ \Comment $v$ decides to join the GMIS
	\EndIf
	\Statex \hspace{3ex}$\blacktriangleright$ \textsc{Part III}: Update the hypergraph
	\If {($state=included$)}
		\State Inform adjacent client nodes about $v$ deciding to join the GMIS
	\EndIf
	\If {(there is an adjacent client $e$ s.t.\ $included$ servers connected to $e$ reaches $t_e$)}
		\State $state\gets excluded$ \Comment $v$ decides to not join the GMIS
	\EndIf
	\State Inform adjacent client nodes about $v$'s decision if it has decided
\EndFor
\end{algorithmic}
\end{small}
\hrule\vspace{1ex}
\caption{\textbf{Pseudocode executed at a node in $\mathcal{H}$ for computing GMIS.}}\label{fig-protocol-linear-gmis-distributed}
\vspace{-3ex}
\end{figure}

\subsection{The Analysis}

The pseudocode clearly indicates the runtime of the algorithm is poly-logarithmic. In this part, we focus on showing the correctness of the algorithm.

To begin with, we show that after part one of each main iteration, an equitable strict subhypergraph containing at least half of the undecided nodes is generated.

\begin{lemma}\label{lemma-linear-gmis-large-eq}
Assume at the beginning of an iteration there are $\hat{n}$ nodes in $\mathcal{H}$ that still have not decided whether to join the GMIS or not. Then, after part one of this iteration, there are at least $\hat{n}/2$ nodes in $active$ state, and they induce an equitable strict subhypergraph.
\end{lemma}

\begin{proof}
During part one, we have an inner loop containing $\Theta(d\cdot\log{\hat{n}})$ iterations. Assume there are $\hat{n}'$ $active$ nodes at the beginning of an inner iteration. If $active$ nodes have not formed an equitable strict subhypergraph yet, then within this inner iteration, each $active$ node $v$ will check whether $d_i(v)>(iu_i/\hat{n}')\cdot\log^4{\hat{n}'}$ for some $1\leq i\leq d-1$. If such $d_i(v)$ exists, then $v$ will set itself as $idle$, and inform adjacent client nodes about this.

We now argue, if at the beginning of an inner iteration, the $\hat{n}'$ $active$ nodes do not form an equitable hypergraph, and by the end of this inner iteration, the updated hypergraph is still not equitable, then by the end of this inner iteration, for some $i$ where $1\leq i\leq d-1$, the number of active threshold $i$ hyperedges decrease by at least a factor of $\log{\hat{n}'}$. To see this, notice that for such an event to happen, there must exist some node $v$ and some $1\leq i\leq d-1$ such that, prior to this inner iteration $d_i(v)\leq(iu_i/\hat{n}')\cdot\log^4{\hat{n}'}$, and after this inner iteration $d'_i(v)>(iu'_i/\hat{n}'')\cdot\log^5{\hat{n}''}$. Notice, according to the definition, we know $d'_i(v)\leq d_i(v)$. If $u_i$ decrease by a factor less than $\log{\hat{n}'}$, then $d'_i(v)>(iu'_i/\hat{n}'')\cdot\log^5{\hat{n}''}>(iu_i/\hat{n}'')\cdot(1/\log{\hat{n}'})\cdot\log^5{\hat{n}''}\geq(iu_i/\hat{n}')\cdot(1/\log{\hat{n}'})\cdot\log^5{\hat{n}'}\geq d_i(v)$, a contradiction.

With the above claim, we argue $\Theta(d\cdot\log{\hat{n}})$ inner iterations are enough to generate an equitable strict subhypergraph. Assume prior to the first inner iteration, we have $x_i$ hyperedges with threshold $i$, where $1\leq i\leq d-1$. Due to Fact \ref{fact-linear-hypergraph-edge-num}, we know $\sum_{i=1}^{d-1}{x_i}\leq{\hat{n}\choose 2}<\hat{n}^2$, which implies $x_i<\hat{n}^2$. After each inner iteration, either we have an equitable strict subhypergraph, or the number of active threshold $i$ hyperedges is decreased by at least a factor of $\log{\hat{n}'}$ for some $i$, where $\hat{n}'$ is the number of active nodes prior to this inner iteration. Notice, once the number of active threshold $i$ hyperedges drops below one for all $i\leq d-1$, the resulting strict subhypergraph must be equitable. On the other hand, for $x_i$ to drop below one, we need at most $O(\log{\hat{n}})$ inner iterations. Hence, the total number of inner iterations we need is at most $O(d\cdot\log{\hat{n}})$.

Lastly, we argue that the equitable strict subhypergraph generated by part one contains at least $\hat{n}/2$ nodes. To see this, notice that for arbitrary $i$, prior to an inner iteration, if there are $\hat{n}'$ $active$ nodes in total, then there are at most $(d/i)\cdot\hat{n}'/\log^4{\hat{n}'}$ nodes satisfying $d_i(v)>(iu_i/\hat{n}')\cdot\log^4{\hat{n}'}$. Hence, during this inner iteration, we set at most $d^2\cdot\hat{n}'/\log^4{\hat{n}'}$ $active$ nodes to $idle$. Since there are only $O(d\cdot\log{\hat{n}})$ iterations, we know after part one, the generated equitable strict subhypergraph contains at least $\hat{n}/2$ nodes.
\end{proof}

In the following discussion, we consider part two and three of the main iteration. Particularly, we show that if the generated equitable strict subhypergraph contains sufficiently many nodes, then after part two and three, with at least constant probability, lots of previously undecided nodes (in the equitable strict subhypergraph) will make up their minds.

We focus on the most involved case in which $\log^8{n'}/n'\leq\hat{a}\leq e^{-6}$.

\begin{lemma}\label{lemma-linear-gmis-remove-nodes-case2}
Assume after part one of an iteration there are $n'$ nodes in the generated equitable strict subhypergraph. Further assume during part two $\log^8{n'}/n'\leq\hat{a}\leq e^{-6}$. Then, after part three of this iteration, with at least some constant probability, at least $\Theta(n'/\log^{8}{n'})$ previously undecided nodes will decide whether to join the GMIS or not.
\end{lemma}

\begin{proof}
Let $\mathcal{H}'=(\mathcal{V}',\mathcal{E}')$ be the generated equitable hypergraph. Let $\mathcal{W}$ be the set of marked (i.e., $elected$) nodes during part two. Let $X$ be a random variable denoting the number of hyperedges $e$ in $\mathcal{H}'$ satisfying $|e\cap\mathcal{W}|=t_e$. The proof relies on two key claims.

\begin{claim}\label{claim-linear-gmis-remove-nodes-case2-claim1}
With at least some constant probability, $X=\Theta(n'/\log^8{n'})$.
\end{claim}

\begin{proof}
For a hyperedge $e$ in the generated equitable hypergraph $\mathcal{H}'$, define $X_e$ to be an indicator random variable taking value one iff $|e\cap\mathcal{W}|=t_e$.

We now calculate $\mathbb{E}(X)$. By linearity of expectation, we know $\mathbb{E}(X)=\sum_{e\in\mathcal{E}'}{\mathbb{E}(X_e)}=\sum_{e\in\mathcal{E'}}{({|e|\choose t_e}\cdot\hat{a}^{t_e}\cdot(1-\hat{a})^{|e|-t_e})}$. Since $|e|\leq d=\Theta(1)$ and $\log^8{n'}/n'\leq\hat{a}\leq e^{-6}$, we know ${|e|\choose t_e}=\Theta(1)$ and $(1-\hat{a})^{|e|-t_e}=\Theta(1)$. Therefore, we know $\mathbb{E}(X)=\Theta(1)\cdot\sum_{e\in\mathcal{E'}}{\hat{a}^{t_e}}=\Theta(1)\cdot\sum_{i=1}^{d-1}{(u_i\cdot\hat{a}^i)}=\Theta(1)\cdot\Theta(n'/\log^8{n'})=\Theta(n'/\log^8{n'})$.

We will show the concentration of $X$ via Chebyshev's inequality, and hence calculate the variance of $X$: $\mathrm{Var}(X)=\mathrm{Var}(\sum_{e\in\mathcal{E'}}{X_e})=\sum_{e\in\mathcal{E'}}{\mathrm{Var}(X_e)}+\sum_{e,e'\in\mathcal{E'};e\cap e'\neq\emptyset}{\mathrm{Cov}(X_e,X_{e'})}$.

Since $\mathrm{Var}(X_e)=\mathbb{E}(X_e^2)-(\mathbb{E}(X_e))^2\leq\mathbb{E}(X_e^2)=\mathbb{E}(X_e)$, we know $\sum_{e\in\mathcal{E'}}{\mathrm{Var}(X_e)}\leq\mathbb{E}(X)$.

On the other hand, notice:

\vspace{-2ex}
\begin{align*}
\sum_{e,e'\in\mathcal{E'};e\cap e'\neq\emptyset}{\mathrm{Cov}(X_{e},X_{e'})} \quad = & \quad \sum_{e,e'\in\mathcal{E'};e\cap e'\neq\emptyset}{(\mathbb{E}(X_{e} X_{e'})-\mathbb{E}(X_{e}) \mathbb{E}(X_{e'}))} \\
\leq & \quad \sum_{e,e'\in\mathcal{E'};e\cap e'\neq\emptyset}{\mathbb{E}(X_{e}X_{e'})}
\end{align*}

Fix two hyperedges $e$ and $e'$ such that $e\cap e'\neq\emptyset$. Since $X_e$ and $X_{e'}$ are indicator random variables, we know $\mathbb{E}(X_eX_{e'})=\mathbb{P}(X_e=1\wedge X_{e'}=1)$. Since $\mathcal{H'}$ is a linear hypergraph, assume $e\cap e'=\{v\}$. (I.e., $e$ and $e'$ overlaps on node $v$.) By the definition of $X_e$ and $X_{e'}$, event ``$X_e=1\wedge X_{e'}=1$'' happens iff one of the two following (disjoint) events happens: (a) $v$ is marked, $t_e-1$ of the $|e|-1$ nodes in $e\backslash\{v\}$ are marked, and $t_{e'}-1$ of the $|e'|-1$ nodes in $e'\backslash\{v\}$ are marked; or (b) $v$ is not marked, $t_e$ of the $|e|-1$ nodes in $e\backslash\{v\}$ are marked, and $t_{e'}$ of the $|e'|-1$ nodes in $e'\backslash\{v\}$ are marked.

Therefore, we can further bound $\sum_{e\cap e'\neq\emptyset}{\mathbb{E}(X_{e}X_{e'})}$:

\vspace{-2ex}
\begin{align*}
\sum_{e\cap e'\neq\emptyset}{\mathbb{E}(X_{e}X_{e'})} = & \sum_{e\cap e'\neq\emptyset}{\left( \hat{a}\cdot{|e|-1\choose t_e-1}\cdot\hat{a}^{t_e-1}\cdot(1-\hat{a})^{|e|-t_e}\cdot{|e'|-1\choose t_{e'}-1}\cdot\hat{a}^{t_{e'}-1}\cdot(1-\hat{a})^{|e'|-t_{e'}} \right)} +\\
\phantom{+} & \sum_{e\cap e'\neq\emptyset}{\left( (1-\hat{a})\cdot{|e|-1\choose t_e}\cdot\hat{a}^{t_e}\cdot(1-\hat{a})^{|e|-t_e-1}\cdot{|e'|-1\choose t_{e'}}\cdot\hat{a}^{t_{e'}}\cdot(1-\hat{a})^{|e'|-t_{e'}-1} \right)}\\
\leq & \frac{2}{\hat{a}}\cdot\sum_{e\cap e'\neq\emptyset}{\left( {|e|\choose t_e}\cdot\hat{a}^{t_e}\cdot{|e'|\choose t_{e'}}\cdot\hat{a}^{t_{e'}} \right)} \leq \frac{2n'}{\log^8{n'}}\cdot\sum_{e\cap e'\neq\emptyset}{\left( \Theta(1)\cdot\hat{a}^{t_e}\cdot\Theta(1)\cdot\hat{a}^{t_{e'}} \right)}\\
\leq & \Theta\left(\frac{n'}{\log^8{n'}}\right)\cdot\sum_{e\in\mathcal{E'}}{\left( \hat{a}^{t_e}\cdot\sum_{e\cap e'\neq\emptyset}{\hat{a}^{t_{e'}}} \right)}
\end{align*}

Next, we need to estimate $\sum_{e\in\mathcal{E'}}{(\hat{a}^{t_e}\cdot\sum_{e\cap e'\neq\emptyset}{\hat{a}^{t_{e'}}})}$ to upper bound $\sum_{e\cap e'\neq\emptyset}{\mathbb{E}(X_{e}X_{e'})}$. Define $\Delta_i=\max_{v\in\mathcal{V'}}{d_i(v)}$. Since $\mathcal{H'}$ is equitable, we have $\Delta_i\leq(iu_i/n')\cdot\log^5{n'}$ for $1\leq i\leq d-1$. Fix a hyperedge $e\in\mathcal{E'}$, we now give an upper bound of $\sum_{e\cap e'\neq\emptyset}{\hat{a}^{t_{e'}}}$. Consider an arbitrary node $v\in e$. For every $1\leq j\leq d-1$, we know $v$ is contained within $d_j(v)\leq\Delta_j$ hyperedges of threshold $j$. Meanwhile, $|e|\leq d$. Hence, $\sum_{e\cap e'\neq\emptyset}{\hat{a}^{t_{e'}}}\leq d\cdot\sum_{j=1}^{d-1}{(\Delta_j\cdot\hat{a}^j)}$.

As a result, we know:

\vspace{-2ex}
\begin{align*}
\sum_{e\cap e'\neq\emptyset}{\mathbb{E}(X_{e}X_{e'})} \quad \leq & \quad \Theta\left(\frac{n'}{\log^8{n'}}\right)\cdot\sum_{e\in\mathcal{E'}}{\left( \hat{a}^{t_e}\cdot\sum_{e\cap e'\neq\emptyset}{\hat{a}^{t_{e'}}} \right)}\\
\leq & \quad \Theta\left(\frac{n'}{\log^8{n'}}\right)\cdot\sum_{e\in\mathcal{E'}}{\left( \hat{a}^{t_e}\cdot d\cdot\sum_{j=1}^{d-1}{(\Delta_j\cdot\hat{a}^j)} \right)}\\
\leq & \quad \Theta\left(\frac{1}{\log^3{n'}}\right)\cdot\sum_{e\in\mathcal{E'}}{\left( \hat{a}^{t_e}\cdot\sum_{j=1}^{d-1}{(u_j\cdot\hat{a}^j)} \right)}\\
\leq & \quad \Theta\left(\frac{1}{\log^3{n'}}\right)\cdot\sum_{e\in\mathcal{E'}}{\left(\hat{a}^{t_e}\cdot\frac{2n'}{\log^8{n'}}\right)}
\end{align*}

Recall that we have previously shown $\mathbb{E}(X)=\Theta(1)\cdot\sum_{e\in\mathcal{E'}}{\hat{a}^{t_e}}=\Theta(n'/\log^8{n'})$. Thus, $\sum_{e\cap e'\neq\emptyset}{\mathbb{E}(X_{e}X_{e'})}\leq \Theta(1/\log^3{n'})\cdot\mathbb{E}(X)\cdot\sum_{e\in\mathcal{E}'}{\hat{a}^{t_e}}=\Theta(1/\log^3{n'})\cdot(\mathbb{E}(X))^2$.

By now, we know $\mathrm{Var}(X)\leq\mathbb{E}(X)+\Theta(1/\log^3{n'})\cdot(\mathbb{E}(X))^2$. Recall $\mathbb{E}(X)=\Theta(n'/\log^8{n'})$, thus $\mathrm{Var}(X)=O((\mathbb{E}(X))^2/\log^3{n'})$. Hence, the claim follows by Chebyshev's inequality.
\end{proof}

\begin{claim}\label{claim-linear-gmis-remove-nodes-case2-claim2}
With at least some constant probability, $\mathcal{H'}$ contains at most $O(n'/\log^{10}{n'})$ pairs of hyperedges $e,e'$ for which $e\cap e'\neq\emptyset$, $|e\cap\mathcal{W}|\geq t_e$, and $|(e'\backslash e)\cap\mathcal{W}|\geq t_{e'}$.
\end{claim}

\begin{proof}
Let $Y$ denote the number of such pairs of hyperedges. Recall we have defined $\Delta_i=\max_{v\in\mathcal{V'}}{d_i(v)}$; and hence know $\Delta_i\leq(iu_i/n')\cdot\log^5{n'}$ for $1\leq i\leq d-1$.

As a result, we can bound $\mathbb{E}(Y)$ as follows:

\vspace{-2ex}
\begin{align*}
\mathbb{E}(Y) \leq & \sum_{e\in\mathcal{E'}}{\left( {|e|\choose t_e}\cdot\hat{a}^{t_e}\cdot\sum_{e\cap e'\neq\emptyset}{\left({|e'|-1\choose t_{e'}}\cdot\hat{a}^{t_{e'}}\right)} \right)} \leq \left[{d\choose d/2}\right]^2\cdot\sum_{e\in\mathcal{E}'}{\left(\hat{a}^{t_e}\cdot\sum_{e\cap e'\neq\emptyset}{\hat{a}^{t_{e'}}}\right)}\\
\leq & \Theta(1)\cdot\sum_{e\in\mathcal{E}'}{\left( \hat{a}^{t_e}\cdot d\cdot\sum_{i=1}^{d-1}{(\Delta_i\cdot\hat{a}^i)} \right)} \leq \Theta\left(\frac{\log^5{n'}}{n'}\right)\cdot\sum_{e\in\mathcal{E}'}{\left( \hat{a}^{t_e}\cdot\sum_{i=1}^{d-1}{(u_i\cdot\hat{a}^{i})} \right)}\\
= & \Theta\left(\frac{1}{\log^3{n'}}\right)\cdot\sum_{e\in\mathcal{E}'}{\hat{a}^{t_e}} = \Theta\left(\frac{1}{\log^3{n'}}\right)\cdot\sum_{i=1}^{d-1}{(u_i\cdot\hat{a}^{i})} = \Theta\left(\frac{n'}{\log^{11}{n'}}\right)
\end{align*}

By Markov's inequality, the claim follows.
\end{proof}

The above two claims show that in each iteration, with at least some constant probability, in $\mathcal{H'}$ there exists a set $\tilde{\mathcal{E}}'$ of hyperedges of cardinality $\Theta(n'/\log^8{n'})$ such that: (a) for each $e\in\tilde{\mathcal{E}}'$, exactly $t_e$ nodes are in $\mathcal{W}$; (b) for $e\in\tilde{\mathcal{E}}'$ and $e'\in\mathcal{E}'$, if $e\cap e'\neq\emptyset$ then $|(e'\backslash e)\cap\mathcal{W}|< t_{e'}$; and (c) for $e\in\tilde{\mathcal{E}}'$ and $e'\in\tilde{\mathcal{E}}'$, there exist $v\in e$ and $v'\in e'$ such that $v\neq v'$ and both $v,v'$ are not in $\mathcal{W}$. Now, notice that (a) and (b) imply that after part three of the iteration, for each hyperedge in $\tilde{\mathcal{E}}'$, at least one node has decided to not be in the GMIS. Moreover, condition (c) guarantees that these nodes are different. Therefore, we have proved the lemma.
\end{proof}

The remaining two cases (namely, $\hat{a}\leq\log^8{n'}/n'$ and $\hat{a}\geq e^{-6}$) are simpler, interested readers can refer to Lemma \ref{lemma-linear-gmis-remove-nodes-case1} and Lemma \ref{lemma-linear-gmis-remove-nodes-case3} in Appendix \ref{appdix-gmis-lemma} for more details.

Finally, we conclude that these lemmas prove the correctness of our algorithm.

\begin{theorem}\label{thm-linear-gmis}
In the CONGEST model, there exists a distributed algorithm that computes a GMIS for constant dimension linear hypergraphs within poly-logarithmic time, w.h.p.
\end{theorem}

\begin{proof}
Lemma \ref{lemma-linear-gmis-large-eq}, \ref{lemma-linear-gmis-remove-nodes-case2}, \ref{lemma-linear-gmis-remove-nodes-case1}, and \ref{lemma-linear-gmis-remove-nodes-case3} tell us: if prior to an outer iteration there are $n'$ undecided nodes, then after this iteration, with at least some constant probability, at least $\Theta(n'/\log^{17}{n'})$ of these nodes will decide, provided that $n'$ is sufficiently large. Since $n'\leq n$, this means after at most some poly-logarithmic (w.r.t.\ $n$) outer iterations, the number of undecided nodes will be reduced to a sufficiently large constant $c_1$, w.h.p.

Now, once the number of undecided nodes is reduced to $c_1$, during part two of an outer iteration, one of the two following scenarios will happen: (a) $p_0\leq\log^8{n'}/n'$, in which case only one node is selected into $\mathcal{W}$; or (b) $e^{-6}\geq p_0>\log^8{n'}/n'$, in which case each of the $n'$ nodes is selected with probability $p_0$. In the first case, the single selected node will decide after this iteration. In the second case, since $n'\leq c_1$ is a constant, we know with at least constant probability only one of the $n'$ nodes will be selected into $\mathcal{W}$, and will decide after this iteration. Therefore, we can conclude when the number of undecided nodes is at most $c_1$, after each iteration, with at least constant probability, at least one node will decide.

At this point, we can claim that after at most some poly-logarithmic (w.r.t.\ $n$) outer iterations, all nodes in $\mathcal{H}$ will decide, w.h.p. Moreover, it is easy to see that the result indeed is a GMIS of $\mathcal{H}$. Combine this with Lemma \ref{lemma-decomp}, and we have proved the theorem.
\end{proof}

\section{Summary and Discussion}

In this paper, we study the problem of efficient computation of MIS and GMIS in linear hypergraphs in the CONGEST model. In particular, we have developed a poly-logarithmic time randomized algorithm for computing an MIS in arbitrary linear hypergraphs. We have then generalized this algorithm and devised a variant that is able to compute a GMIS in constant dimension linear hypergraphs, again in poly-logarithmic time.

To the best of our knowledge, this is the first work that defines the GMIS problem and devises non-trivial algorithms for computing it. We believe this problem deserves further investigation. On the one hand, it can potentially model many real-world problems that involve multi-party interactions; on the other hand, it is also a challenging symmetry breaking problem and solving it efficiently seems to require the development of novel techniques.

A natural question to ask is how to efficiently compute GMIS for linear hypergraphs with super-constant dimension, in the CONGEST model? (For the LOCAL model, recall that Theorem \ref{thm-gmis-local} already gives the answer.) Why does an algorithm (or, the techniques behind it) that can solve MIS for arbitrary dimension linear hypergraphs stops at constant dimension for GMIS? It turns out there are several difficulties. To begin with, for the key parameter $\hat{a}$, in the GMIS setting, instead of our current definition, the most natural one should actually be $\sum_{i=2}^{d}\sum_{j=1}^{i-1}{i\choose j}\cdot u_{i,j}\cdot\hat{a}^j\cdot(1-\hat{a})^{i-j}=\Theta(n'/\log^{\Theta(1)}{n'})$, where $u_{i,j}$ is the number of size $i$ hyperedges with threshold $j$. However, this definition would break the proof of Lemma \ref{lemma-linear-gmis-remove-nodes-case2}. Particularly, the analysis for Claim \ref{claim-linear-gmis-remove-nodes-case2-claim2} is no longer valid. On the other hand, once we introduce the notion of $u_{i,j}$, the definition for equitable hypergraph also needs to be adjusted: in the GMIS setting, $\mathcal{H}$ is equitable if it contains not too many nodes, or for each node $v$, for each $i$ where $2\leq i\leq d$, for each $j$ where $1\leq j\leq i-1$, it holds that $d_{i,j}(v)\leq(i u_{i,j}/n)\cdot\log^{\Theta(1)}{n}$. Unfortunately, this definition could greatly increase the time complexity of the equitable subhypergraph generation algorithm: for given $i$ and $j$, the value of $d_{i,j}(v)$ is not necessarily monotonically decreasing over multiple iterations. To summarize, we have the feeling that GMIS might be fundamentally harder than MIS, and that obtaining more general solutions might require non-trivial novel algorithmic techniques.

\clearpage
\bibliographystyle{plainurl}
\bibliography{ref}

\clearpage
\appendix
\section*{Appendix}

\section{More Details on Obtaining Required Parameters}\label{appdix-obtain-para}

Here we describe how to obtain $n'$ and $u_i(\mathcal{H}')$ in more detail.

\begin{itemize}
	\item \textbf{Leader election.} The first building block is leader election. Specifically, in the simple graph which is the server-client representation of $\mathcal{H}$, we need to elect a single node $v$ as the leader, so that every node (including $v$ itself) knows $v$ is the leader.
	
	Assume each node has a unique identity, we will elect the one with the largest identity to be the leader.\footnote{Notice, we have already assumed each server node (i.e., each node in $\mathcal{H}$) has a unique identity. In case client nodes do not have unique identities, they can randomly sample one from a sufficiently large pool (say, a pool of size $\Omega(n^3)$).} To achieve this, we only need to let each node broadcast the largest identity it has seen for a certain number of slots. (In the first slot, each node broadcasts its own identity.) Since the server-client representation of $\mathcal{H}$ has diameter $O(\log{n})$, we know we can elect a leader in $O(\log{n})$ time slots.
	
	\item \textbf{Tree construction.} The second building block is spanning tree construction. More specifically, we want to build a BFS tree on the server-client representation of $\mathcal{H}$. The procedure is based on the distributed implementation of Bellman and Ford's algorithm (see, e.g., Section 5.3 of \cite{peleg00}).
	
	Particularly, to construct this tree, we first run the leader election procedure described above, and let the leader be the root. In each time slot, each node that is already in the tree sends its distance to the root to all of its neighbors, along with its identity. (So in the first time slot, only the root sends a message.) Each node that receives a message and is not in the tree yet will add itself to the tree. Each such node will also know its parent in the tree, and its distance to the root. Since the server-client representation of $\mathcal{H}$ has diameter $O(\log{n})$, we know this tree can be constructed in $O(\log{n})$ time slots.
	
	\item \textbf{Compute $n'$.} With a BFS tree, we can now count the number of nodes in $\mathcal{H}'$, and let each node in the server-client representation of $\mathcal{H}'$ be aware of this count.
	
	To achieve this goal, we first aggregate the count from the leaves to the root, and then let the root broadcast the count. More specifically, once tree construction is done, the following step will be repeated for $\Theta(\log{n})$ times: in each time slot, each node $v$ in the tree will send to its parent the number of nodes it knows that are in $V(\mathcal{H}')$ and are contained within the subtree rooted at $v$, provided that they have not been counted previously. Effectively, this means in the first time slot each server node in the tree (i.e., each node in $\mathcal{H}'$) will send one to its parent; and after the first time slot, each node in the tree will know how many of its one-hop children are in $V(\mathcal{H}')$. In general, after $k$ slots, for each node in the tree, it will know among its $k$-hop descendants, how many are in $V(\mathcal{H}')$. Since the tree has depth $O(\log{n})$, the root will know the number of nodes in $\mathcal{H}'$ in $O(\log{n})$ time. Once the root knows the count, it can broadcast this count, which takes an additional $O(\log{n})$ time. To sum up, we can conclude that each node in the server-client representation of $\mathcal{H}'$ will know $n'$ in $O(\log{n})$ time slots.
	
	\item \textbf{Compute $u_i$.} Finally, we discuss how to count $u_i$ for $2\leq i\leq\log{n'}$. Particularly, for each $2\leq i\leq\log{n'}$, we need each node in the server-client representation of $\mathcal{H}'$ to know $u_i$. To accomplish this task, we need $\log{n}$ iterations, each of which is similar to the node counting procedure described above: we reuse the tree constructed for node counting, and aggregate the count (of $u_i$) from the leaves to the root, and finally let the root broadcast $u_i$ to all other nodes. The correctness argument is similar to the one for node counting, and the total time consumption will be $O(\log^2{n})$.
\end{itemize}

\section{Omitted Lemmas for the Analysis of the GMIS Algorithm}\label{appdix-gmis-lemma}

\begin{lemma}\label{lemma-linear-gmis-remove-nodes-case1}
Assume after part one of an iteration there are $n'$ nodes in the generated equitable strict subhypergraph. Further assume during part two $\hat{a}\leq\log^8{n'}/n'$. Then, after part three of this iteration, at least $\Theta(n'/\log^{17}{n'})$ previously undecided nodes will decide whether to join the GMIS or not.
\end{lemma}

\begin{proof}
When $\hat{a}\leq\log^8{n'}/n'$, we have $\sum_{i= 2}^{d-1}{(u_i\cdot\hat{a}^i)}\leq\sum_{i=2}^{d-1}{u_i\cdot\hat{a}^2}\leq(\log^{16}{n'}/(n')^2)\cdot\sum_{i=2}^{d-1}{u_i}\leq(\log^{16}{n'}/(n')^2)\cdot(n')^2=\log^{16}{n'}$. Hence, $u_1=(1/\hat{a})\cdot(\sum_{i=1}^{d-1}u_i\hat{a}^i-\sum_{i=2}^{d-1}{u_i\hat{a}^i})\geq(1/\hat{a})\cdot(n'/\log^8{n'}-\log^{16}{n'})\geq(n'/\log^8{n'})\cdot(n'/\log^9{n'})=(n')^2/\log^{17}{n'}$. As a result, for the node $v$ which maximizes $d_1(v)$, we have $d_1(v)\geq u_1/n'\geq n'/\log^{17}{n'}$.

Now, notice that during part three, for the $d_1(v)$ threshold one hyperedges that contain $v$, the other nodes in each of these hyperedges will decide to not be in the GMIS (since $v$ is already in the GMIS, and each such hyperedge must have size at least two). Hence, at least $d_1(v)$ nodes will make up their minds.
\end{proof}

\begin{lemma}\label{lemma-linear-gmis-remove-nodes-case3}
Assume after part one of an iteration there are $n'$ nodes in the generated equitable strict subhypergraph. Further assume during part two $\hat{a}\geq e^{-6}$. Then, after part three of this iteration, with at least some constant probability, at least $\Theta(n')$ previously undecided nodes will decide whether to join the GMIS or not.
\end{lemma}

\begin{proof}
Since $\hat{a}\geq e^{-6}$, we know each node in the generated equitable hypergraph $\mathcal{H'}$ will be selected with probability $e^{-6}$. Hence, we know w.h.p.\ w.r.t.\ $n'$, $\Theta(n')$ nodes will be selected into $\mathcal{W}$. That is, with at least some constant probability, $|\mathcal{W}|=\Theta(n')$.

On the other hand, in expectation, the number of nodes in $\mathcal{W}$ that belong to some hyperedge that exceeds the threshold is upper bounded by $\sum_{i=1}^{d-1}{\sum_{j=i+1}^{d}{u_{i,j}\cdot j\cdot{j\choose i}\cdot p_0^i}}\leq\sum_{i=1}^{d-1}{\sum_{j=i+1}^{d}{u_{i,j}\cdot d\cdot{d\choose d/2}\cdot p_0^i}}\leq d\cdot e^d\cdot\sum_{i=1}^{d-1}{\sum_{j=i+1}^{d}{u_{i,j}\cdot p_0^i}}=d\cdot e^d\cdot\sum_{i=1}^{d-1}{u_i\cdot p_0^i}\leq d\cdot e^d\cdot\sum_{i=1}^{d-1}{u_i\cdot \hat{a}^i}\leq 2d\cdot e^d\cdot n'/\log^8{n'}=O(n'/\log^7{n'})$. Here, $p_0=e^{-6}$, and $u_{i,j}$ is the number of size $j$ hyperedges with threshold $i$. Hence, by a Markov's inequality, we know with at least some constant probability, after part two of the iteration, we can find a generalized independent set of size $\Theta(n')$. Moreover, the nodes in this set will decide to join the GMIS by the end of this iteration.
\end{proof}

\end{document}